\documentclass[11pt,a4paper]{article}

\usepackage[utf8]{inputenc}
\usepackage[T1]{fontenc}

\usepackage[margin=2.5cm]{geometry}

\usepackage{amsmath,amssymb,amsthm}
\usepackage{mathtools}
\usepackage{bm}

\usepackage{graphicx}
\usepackage{booktabs}
\usepackage{tabularx}
\usepackage{multirow}
\usepackage{threeparttable}
\usepackage{array}
\usepackage{dcolumn}

\usepackage{hyperref}
\usepackage[round,authoryear]{natbib}
\usepackage{url}

\usepackage{setspace}
\usepackage{xcolor}
\usepackage{lineno}

\graphicspath{{./images/}}

\newtheorem{theorem}{Theorem}[section]
\newtheorem{proposition}[theorem]{Proposition}

\theoremstyle{definition}

\theoremstyle{remark}
\newtheorem{remark}[theorem]{Remark}

\begin{document}

\title{\textbf{Parametric ROC Analysis and Optimal Cutoff Selection
       under Scale Mixtures of Skew-Normal Distributions:
       A Decision-Theoretic Framework with Asymptotic Inference}}

\author{
  Renato de Paula\textsuperscript{1,2}\thanks{Corresponding author. E-mail: \texttt{rrpaula@ciencias.ulisboa.pt}}
  \and
  Helena Mouriño\textsuperscript{1,2}
  \and
  Tiago Dias Domingues\textsuperscript{1,2}
}

\date{
  \textsuperscript{1}Departamento de Ciências Matemáticas,
  Faculdade de Ciências, Universidade de Lisboa, Lisboa, Portugal\\
  \textsuperscript{2}Centro de Estatística e Aplicações (CEAUL),
  Faculdade de Ciências, Universidade de Lisboa, Lisboa, Portugal
}

\maketitle

\begin{abstract}
Continuous biomarkers must often be converted into binary diagnostic
decisions through the choice of a single classification threshold.
In practice, this threshold should reflect not only the biomarker's
discriminatory ability, as summarised by the ROC curve, but also disease
prevalence and the relative costs of false-positive and false-negative
errors. The widely used Youden index implicitly assumes equal costs and
equal prevalences, assumptions that are often unrealistic. In addition,
biomarkers from serological and immunological studies frequently exhibit
skewness and heavy tails, which limit the adequacy of Gaussian-based
ROC models.

We develop a unified parametric framework for ROC analysis and optimal
cutoff selection under the family of scale mixtures of skew-normal
(SMSN) distributions, including the skew-normal and skew-$t$ models.
The ROC curve and AUC are estimated using plug-in maximum-likelihood
methods based on separate-group fits. The optimal cutoff is defined by
minimising a weighted misclassification risk that incorporates disease
prevalence and asymmetric costs, leading to a likelihood-ratio equation
that generalises the Youden criterion. Under a monotone likelihood ratio
condition, we establish existence, uniqueness, and global optimality of
the cutoff, derive a first-order expansion of its departure from the
Youden solution, and obtain consistency, asymptotic normality, and a
closed-form plug-in variance estimator. A key component of the variance
is the local slope of the estimating equation at the optimal threshold,
which we interpret as a local identifiability diagnostic.

Monte Carlo simulations across six scenarios confirm the accuracy of
the asymptotic approximation and the nominal coverage of Wald confidence
intervals. An analysis of SARS-CoV-2 serological data shows that the
proposed cutoff can differ substantially from the Youden threshold and
reduce the estimated misclassification risk by up to 63\% under
asymmetric decision settings.
\end{abstract}


\bigskip

\section{Introduction}
\label{sec:introduction}

The selection of a diagnostic threshold for a continuous biomarker is a
routine but consequential clinical decision. In screening programmes for
infectious diseases, in stratifying patients by risk of progression, and
in monitoring treatment response, laboratory tests must ultimately yield
a binary classification: positive or negative. The threshold that
determines this classification has direct implications for patient care:
setting it too low increases the number of unnecessary confirmatory
procedures and associated costs; setting it too high misses cases that
would benefit from early intervention. Despite its clinical importance,
threshold selection is often performed using criteria that implicitly
assume that false positives and false negatives are equally costly and
that the disease is as prevalent as its absence, which rarely holds in
practice. In a serological survey conducted in a low-prevalence
population, for instance, a symmetric criterion such as the Youden index
will generate an excess of false positives that may overwhelm
confirmatory testing capacity; conversely, in a high-risk cohort, a
threshold optimised without accounting for the elevated cost of missed
diagnoses may substantially underdetect cases.

The receiver operating characteristic (ROC) curve is the standard
framework for evaluating the discriminatory capacity of a continuous
biomarker across all possible thresholds
\citep{metz1978basic,pepe2003statistical,zweig1993receiver}. By
plotting sensitivity against one minus specificity as the threshold
varies, it summarises the complete trade-off between the two types of
classification error and provides a threshold-free summary of accuracy
through the area under the curve (AUC), which equals the probability
that a randomly selected diseased individual has a higher marker value
than a randomly selected healthy individual
\citep{bamber1975area,hanley1982meaning}. The ROC framework has been
studied under both nonparametric \citep{delong1988comparing} and
parametric assumptions \citep{metz1978basic,pepe2003statistical}, with
the binormal model being the most widely used parametric approach
\citep{dorfman1969maximum,metz1998proper}.

Once a biomarker has been evaluated through its ROC curve, clinical
implementation requires selecting a single operating point. The most
widely used selection criterion is the Youden index
\citep{youden1950index}, which maximises the sum of sensitivity and
specificity, and corresponds geometrically to the point on the ROC curve
where the tangent has unit slope \citep{metz1978basic}. Statistical
methods for estimating the Youden cutoff under normality were developed
by \citet{fluss2005estimation} and extended to more general settings by
\citet{schisterman2005optimal} and \citet{perkins2006inconsistency}.
However, the Youden criterion implicitly assigns equal costs to the two
error types and treats the disease as equally prevalent to its absence.
\citet{mcintosh2002combining} established that, under a weighted
misclassification risk incorporating disease prevalence and asymmetric
costs, the optimal threshold corresponds to the point on the ROC curve
whose tangent slope equals the cost-prevalence ratio, recovering the
Youden solution only when this ratio equals one. In epidemiological
settings where disease is rare, or where the consequences of the two
error types differ substantially, the Youden threshold can be far from
optimal and may yield considerably higher expected costs than a properly
calibrated decision-theoretic cutoff.

A parametric approach to ROC analysis and cutoff estimation requires
specifying distributional models for the biomarker in each group. In
serological and immunological applications, antibody concentrations are
characteristically right-skewed and heavy-tailed: a small proportion of
individuals exhibit extreme values that are not well captured by
Gaussian or log-normal models, and these extreme values may be
diagnostically informative \citep{lachos2010likelihood}. Parametric ROC
models under log-normal, gamma, and beta distributions have been
proposed, but these families may be insufficient for data that are
simultaneously skewed and heavy-tailed, as commonly observed in
serological studies of infectious disease.

The family of scale mixtures of skew-normal (SMSN) distributions
provides a rich and flexible framework that encompasses the skew-normal
\citep{azzalini1985class}, skew-$t$ \citep{azzalini2003distributions},
skew-slash \citep{wang2004skew}, and skew-contaminated normal
distributions as special cases. The family accommodates right or left
skewness and varying degrees of tail heaviness, and reduces to the
Gaussian family when the skewness parameter is zero, making it a
natural generalisation of the binormal ROC model.
\citet{azzalini2014skew} provide a comprehensive treatment of the SMSN
family, and \citet{arevalillo2012maximum} derive explicit asymptotic
results for maximum likelihood estimation in the skew-normal and
skew-$t$ submodels. Applications to serological and diagnostic data
have demonstrated substantial advantages of SMSN models over Gaussian
alternatives \citep{lachos2010likelihood,diasdomingues2024classification,
domingues2024analysis}. Despite this, a systematic treatment of ROC
curve estimation, AUC inference, and optimal cutoff selection within the
SMSN family that integrates decision-theoretic optimality with formal
asymptotic inference has not, to our knowledge, been developed. The
present paper addresses this gap through three integrated contributions.

\smallskip

Initially, we develop a fully parametric ROC
model within the SMSN family, deriving the ROC curve and AUC in closed
form as smooth functionals of the SMSN parameters, and providing
consistent plug-in estimators based on separate maximum-likelihood fits
in the two groups.

\smallskip

Next, we establish a rigorous decision-theoretic framework for optimal cutoff selection. The optimal
threshold is characterised as the unique solution to a likelihood-ratio
equation that generalises the Youden criterion to settings with
asymmetric misclassification costs and unequal disease prevalence.
Under a monotone likelihood ratio condition, we prove existence,
uniqueness, and global optimality of the proposed cutoff, and derive a
first-order expansion that quantifies the displacement of the optimal
cutoff from the Youden solution as a function of the cost-prevalence
ratio and the local curvature of the log-likelihood ratio at the Youden
point.

\smallskip

Finally, we derive the joint asymptotic
distribution of the plug-in estimators of the ROC curve, AUC, and
optimal cutoff via the implicit function theorem and the multivariate
delta method. A key structural feature of the resulting asymptotic
variance is that its denominator coincides with the local slope of the
estimating equation at the optimal threshold, a quantity we term the
local identifiability diagnostic. This quantity measures how
sharply the decision boundary is defined, and it provides a practical
tool for assessing the reliability of the estimated cutoff before
reporting it clinically.

\smallskip

The theoretical results are validated through a Monte Carlo simulation
study covering six scenarios spanning skew-normal and skew-$t$
distributions with varying degrees of class separation, cost-prevalence
asymmetry, and tail heaviness. The methodology is illustrated through an
application to SARS-CoV-2 serological data, in which the proposed
cutoff reduces the estimated misclassification risk by up to 63\%
relative to the Youden threshold under asymmetric decision settings.

The remainder of the paper is organised as follows.
Section~\ref{sec:smsn} introduces the SMSN family and its principal
submodels. Section~\ref{sec:setup} establishes the notation and model
structure. Section~\ref{sec:roc} develops the parametric ROC model and
AUC estimator. Section~\ref{sec:cutoff} presents the decision-theoretic
framework, its theoretical properties, and asymptotic inference.
Section~\ref{sec:simulation} reports the simulation study.
Section~\ref{sec:application} presents the application to SARS-CoV-2
serological data. Section~\ref{sec:discussion} concludes.

\section{Scale mixtures of skew-normal distributions}
\label{sec:smsn}

\subsection{Definition and stochastic representation}

Let $Z$ be a continuous random variable.
We say that $Z$ follows a scale mixture of skew-normal
(SMSN) distribution with location parameter $\xi \in \mathbb{R}$,
scale parameter $\omega > 0$, shape parameter $\alpha \in \mathbb{R}$,
and mixing distribution $H(\cdot;\nu)$ indexed by $\nu$,
written $Z \sim \mathrm{SMSN}(\xi,\omega,\alpha;H)$, if it admits the
stochastic representation
\[
Z=\xi+\frac{W}{\sqrt{U}},
\]
where $U>0$ almost surely has distribution $H(\cdot;\nu)$,
independent of $W \sim \mathrm{SN}(0,\omega^2,\alpha)$. Conditionally on $U=u$, one has
\[
Z \mid U=u \sim \mathrm{SN}\!\left(\xi,\frac{\omega^2}{u},\alpha\right),
\qquad u>0.
\]
Hence the marginal density of $Z$ is
\begin{equation}
\label{eq:smsn_density}
f_Z(z;\xi,\omega,\alpha,\nu)
=
\int_0^{\infty}
\frac{2\sqrt{u}}{\omega}\,
\phi\!\left(\frac{z-\xi}{\omega/\sqrt{u}}\right)
\Phi\!\left(\frac{\alpha(z-\xi)}{\omega/\sqrt{u}}\right)
\, dH(u;\nu),
\end{equation}
$z\in\mathbb{R}$, where $\phi$ and $\Phi$ denote the standard normal density and distribution function, respectively.
When $\alpha=0$, equation~\eqref{eq:smsn_density} reduces to the
density of a symmetric scale mixture of normals \citep{andrews1974scale}.

\subsection{Particular cases: skew-normal and skew-t distributions}

The two submodels used throughout this paper are the following.\\

\noindent
\textbf{Skew-normal (SN) distribution}:
If $U\equiv 1$, then
\[
Z \sim \mathrm{SN}(\xi,\omega^2,\alpha),
\]
with density
\[
f_Z(z)
=
\frac{2}{\omega}\,
\phi\!\left(\frac{z-\xi}{\omega}\right)
\Phi\!\left(\frac{\alpha(z-\xi)}{\omega}\right),
\qquad z\in\mathbb{R}.
\]
The parameter $\alpha>0$ ($\alpha<0$) produces right-skewness
(left-skewness), whereas $\alpha\!=\!0$ reduces to the Gaussian model, with mean $\xi$ and variance $\omega^2$,
$\mathcal{N}(\xi,\omega^2)$. The skew-normal distribution was introduced by \citet{azzalini1985class}.\\

\noindent
\textbf{Skew-$t$ (ST) distribution}:
If $U=V/\nu$ with $V\sim\chi^2_\nu$, then
\[
Z \sim \mathrm{ST}(\xi,\omega^2,\alpha,\nu),
\qquad \nu>2,
\]
in the Azzalini--Capitanio parametrisation \citep{azzalini2003distributions}.
The skew-$t$ extends the skew-normal model by allowing heavier tails,
and converges to the skew-normal distribution as $\nu\to\infty$.
When $\alpha=0$, it reduces to the symmetric Student-$t$ model
with location $\xi$, scale $\omega$, and $\nu$ degrees of freedom. Throughout this paper, we use the unconstrained parametrisation
\begin{equation}
\label{eq:parametrisation}
\theta =
(\xi,\log\omega,\alpha)
\quad\text{(SN)},
\qquad
\theta =
(\xi,\log\omega,\alpha,\log(\nu-2))
\quad\text{(ST)},
\end{equation}
which ensures $\omega>0$ and, in the skew-$t$ case, $\nu>2$
throughout numerical optimisation.

\section{Diagnostic model and unified notation}
\label{sec:setup}

The following notation and model structure are used throughout
Sections~\ref{sec:roc} and~\ref{sec:cutoff} without further
redefinition.

Let $D \in \{0,1\}$ denote the disease status of an individual,
where $D = 0$ corresponds to the non-diseased (seronegative)
individual and $D = 1$ to the diseased (seropositive) individual.
Let $X$ be a continuous biomarker measured on each individual,
and denote by $\pi_k = P(D = k)$, $k = 0, 1$, the class
prevalences, with $\pi_0, \pi_1 > 0$ and $\pi_0 + \pi_1 = 1$. Conditionally on $D = k$, the biomarker follows a member of
the SMSN family (Section~\ref{sec:smsn}):
\[
X \mid D = k \;\sim\; \mathrm{SMSN}(\xi_k, \omega_k, \alpha_k; H_k),
\quad k = 0, 1.
\]
We denote the cumulative distribution function (CDF) and
probability density function (PDF) of $X \mid D = k$ by
$F_k(\cdot;\theta_k)$ and $f_k(\cdot;\theta_k)$, respectively,
where
\[
\theta_k \in \Theta_k \subset \mathbb{R}^{p_k}
\]
is the parameter vector of group $k$ in the unconstrained
parametrisation~\eqref{eq:parametrisation},
with $p_k = 3$ for the SN and $p_k = 4$ for the ST submodel.
The joint parameter vector is
\[
\theta = (\theta_0^\top, \theta_1^\top)^\top
\;\in\;
\Theta = \Theta_0 \times \Theta_1 \subset \mathbb{R}^{p_0+p_1} = \mathbb{R}^{p}.
\]

For a threshold $c$ in the admissible interval
$\mathcal{C} = [a,b] \subset \mathbb{R}$, define the
classification rule
\begin{equation}
\label{eq:decision-rule}
\delta_c(x)
=
\begin{cases}
1, & x > c, \\
0, & x \leq c,
\end{cases}
\end{equation}
so that individuals with $X \leq c$ are classified as
seronegative and those with $X > c$ as seropositive.
The interval $\mathcal{C}$ is interpreted as a compact decision
region covering the clinically plausible range of the marker;
it is chosen to contain the effective support of both fitted
distributions.

\subsection{False positive and true positive fractions}
Under the classification rule~\eqref{eq:decision-rule},
the false positive fraction (FPF) and true positive fraction
(TPF) at threshold $c$ are
\begin{equation}
\label{eq:fpf-tpf}
\mathrm{FPF}(c;\theta_0)
= 1 - F_0(c;\theta_0),
\qquad
\mathrm{TPF}(c;\theta_1)
= 1 - F_1(c;\theta_1).
\end{equation}

\noindent
Equivalently, the sensitivity and specificity at threshold $c$ are
\begin{equation}
\label{eq:se-sp}
\begin{split}
\mathrm{Se}(c;\theta_1) &= \mathrm{TPF}(c;\theta_1) = 1 - F_1(c;\theta_1), \\
\mathrm{Sp}(c;\theta_0) &= 1 - \mathrm{FPF}(c;\theta_0) = F_0(c;\theta_0).
\end{split}
\end{equation}

\subsection{Misclassification costs}
A false negative (FN) occurs when $D = 1$ but $\delta_c(X) = 0$;
a false positive (FP) occurs when $D = 0$ but $\delta_c(X) = 1$.
We assign positive costs $\lambda_0 > 0$ and $\lambda_1 > 0$
to FPs and FNs, respectively, reflecting the potentially
asymmetric consequences of the two error types.

\begin{remark}[Prevalence versus sampling proportions]
\label{rem:prevalence}
\normalfont
The prevalences $\pi_0$ and $\pi_1$ refer to the target
\emph{population}, not to the sampling proportions $n_0/n$
and $n_1/n$.
In case-control studies, the latter do not estimate $\pi_k$
and must not be used as substitutes.
Throughout this paper, $\pi_k$ are treated as known external
quantities, specified from epidemiological evidence or
decision-theoretic considerations. The extension to estimated
prevalences is discussed in Section~\ref{sec:discussion}.
\end{remark}


\section{Parametric ROC model under SMSN distributions}
\label{sec:roc}

\subsection{The ROC curve}

Using the FPF and TPF defined in~\eqref{eq:fpf-tpf}, the
parametric ROC curve under the SMSN model is obtained by
eliminating $c$: setting $r = \mathrm{FPF}(c;\theta_0)$
and inverting,
$c = F_0^{-1}(1-r;\theta_0) =: Q_0(1-r;\theta_0)$,
one obtains
\begin{equation}
\label{eq:roc-smsn}
\mathrm{ROC}_{\mathrm{SMSN}}(r;\theta)
= 1 - F_1\!\bigl(Q_0(1-r;\theta_0);\theta_1\bigr),
\quad r \in (0,1).
\end{equation}
Equation~\eqref{eq:roc-smsn} generalises the binormal ROC
model by replacing Gaussian CDFs with SMSN CDFs;
since neither $F_k$ nor $Q_0$ has a closed form in general,
both are evaluated numerically via the
\texttt{psn}/\texttt{pst} and \texttt{qsn}/\texttt{qst}
functions of the \texttt{sn} package \citep{azzalini2020package}.

A key identity, which holds for any pair of absolutely
continuous distributions, connect the
local geometry of the ROC curve to the likelihood ratio:
\begin{equation}
\label{eq:roc-slope}
\frac{d\,\mathrm{TPF}}{d\,\mathrm{FPF}}\bigg|_c
= \frac{f_1(c;\theta_1)}{f_0(c;\theta_0)}
=: \Lambda(c;\theta).
\end{equation}
This identity plays a central role in both the Youden index
(Section~\ref{subsec:youden}) and the decision-theoretic
optimal cutoff (Section~\ref{sec:cutoff}).


\subsection{The Youden index}
\label{subsec:youden}

The Youden index \citep{youden1950index} is the most widely
used threshold selection criterion derived directly from the
ROC curve.
For a threshold $c \in \mathcal{C}$, it is defined as the
vertical distance between the ROC curve and the diagonal
corresponding to a non-informative classifier. Using the sensitivity $\mathrm{Se}(c;\theta_1)$ and specificity
$\mathrm{Sp}(c;\theta_0)$ defined in
equation~\eqref{eq:se-sp}, the Youden index at threshold $c$ is
\begin{equation}
\label{eq:youden-index}
J(c;\theta)
= \mathrm{Se}(c; \theta_1) + \mathrm{Sp}(c;\theta_0) - 1
= F_0(c;\theta_0) - F_1(c;\theta_1).
\end{equation}
The population Youden cutoff is the maximiser of
$J(\cdot;\theta)$ over $\mathcal{C}$:
\begin{equation}
\label{eq:youden-cutoff}
c_Y(\theta)
= \arg\max_{c \in \mathcal{C}}\,J(c;\theta).
\end{equation}

\noindent
Differentiating $J(c;\theta)$ with respect to $c$ gives
\[
\partial_c J(c;\theta)
= f_0(c;\theta_0) - f_1(c;\theta_1),
\]
so any interior maximiser satisfies
$f_1(c_Y;\theta_1) = f_0(c_Y;\theta_0)$, equivalently
\begin{equation}
\label{eq:youden-condition}
\Lambda(c_Y(\theta);\theta) = 1.
\end{equation}
By equation~\eqref{eq:roc-slope}, condition
\eqref{eq:youden-condition} states that the Youden cutoff
corresponds to the point on
$\mathrm{ROC}_{\mathrm{SMSN}}$ whose tangent has
unit slope \citep{metz1978basic, pepe2003statistical}.

The plug-in Youden estimator under the SMSN model is
\begin{equation}
\label{eq:youden-estimator}
\hat{c}_Y
= \arg\max_{c \,\in \,\mathcal{C}}\,
  J(c;\hat\theta)
= \arg\max_{c \, \in \, \mathcal{C}}\,
  \bigl[\hat F_0(c;\hat\theta_0) - \hat F_1(c;\hat\theta_1)\bigr],
\end{equation}
where $\hat\theta_0$ and $\hat\theta_1$ are the separate
maximum likelihood estimators of the two groups.
For comparison, the nonparametric empirical Youden
estimator is
\begin{equation}
\label{eq:youden-empirical}
\hat{c}_Y^{\,\mathrm{emp}}
= \arg\max_{c \in \mathcal{C}}
  \left\{
    \frac{1}{n_1}\sum_{j=1}^{n_1}\mathbf{1}(X_{1j} > c)
    +
    \frac{1}{n_0}\sum_{i=1}^{n_0}\mathbf{1}(X_{0i} \leq c)
    - 1
  \right\},
\end{equation}
which provides a useful nonparametric benchmark in
simulation studies and empirical analyses.



\subsection{Area under the ROC curve}
\label{subsec:auc}

The AUC summarises global discriminatory performance by
integrating the ROC curve over all operating points.
It equals $P(X_1 > X_0)$ where
$X_k \sim F_k(\cdot;\theta_k)$ are independent, and assume the representation
\begin{equation*}
\label{eq:auc-smsn}
\mathrm{AUC}_{\mathrm{SMSN}}(\theta)
= \int_0^1 \mathrm{ROC}_{\mathrm{SMSN}}(r;\theta)\,dr
= \int_0^1
  \bigl[1 - F_1\bigl(Q_0(1-r;\theta_0);\theta_1\bigr)\bigr]\,dr,
\end{equation*}
obtained via equation~\eqref{eq:roc-smsn}.
Equivalently, the change of variable $u = 1-r$ gives
\[
\mathrm{AUC}_{\mathrm{SMSN}}(\theta)
= \int_0^1
  \bigl[1 - F_1\bigl(Q_0(u;\theta_0);\theta_1\bigr)\bigr]\,du,
\]
which coincides with the representation $P(X_1 > X_0)$
obtained via $u = F_0(x;\theta_0)$ \citep{bamber1975area}.
The plug-in estimator replaces $\theta$ by $\hat\theta$ and
approximates the integral by the trapezoidal rule:
\begin{equation*}
\label{eq:auc-numerical}
\widehat{\mathrm{AUC}}(\theta)_{\mathrm{SMSN}}
= \frac{1}{h}\!\left[
    \frac{1}{2}
    + \sum_{i=1}^{h-1}
      \Bigl(1 - \hat F_1\bigl(
        \hat Q_0(u_i;\hat\theta_0);\hat\theta_1
      \bigr)\Bigr)
  \right],
\quad u_i = \frac{i}{h}.
\end{equation*}
\noindent
The trapezoidal rule with $h$ equally spaced points $u_i = i/h$,
$i = 1,\ldots,h-1$, has step size $\delta = 1/h$ and approximation
error $O(\delta^2) = O(h^{-2})$ \cite{davis1984methods}.
Setting $h = 1000$ yields a step size of $10^{-3}$ and an approximation
error of order $10^{-6}$, which is negligible relative to the estimation
uncertainty of $\hat{\theta}$ in all SMSN submodels considered here.

\begin{remark}[Three complementary summaries of the ROC curve]
\label{rem:three-summaries}
\normalfont
The Youden index, the AUC, and the optimal cutoff provide
three complementary perspectives on diagnostic performance,
each operating at a different level of aggregation.
The AUC integrates over the entire ROC curve and measures
global discrimination without committing to any operating
point.
The Youden index selects the single operating point where
sensitivity and specificity are jointly maximised, implicitly
assuming equal costs and prevalences.
The optimal cutoff (Section~\ref{sec:cutoff}) selects the
operating point that minimises the expected misclassification
cost for a specified cost-prevalence configuration, of which
the Youden solution is the symmetric special case.
In the SMSN framework, all three are estimated as smooth
functionals of the MLE $\hat\theta$, enabling unified
asymptotic inference via the delta method
(Theorem~\ref{thm:asymptotics}).
\end{remark}


\section{Optimal cutoff selection under misclassification costs}
\label{sec:cutoff}

\subsection{Decision-theoretic formulation}

Building on the notation introduced in Section~\ref{sec:setup},
consider the classification rule $\delta_c(x)=\mathbf{1}(x>c), c\in\mathcal C=[a,b]$, under which an individual is classified as seropositive when the biomarker exceeds the threshold $c$ and as seronegative otherwise.

Let $\lambda_1>0$ denote the cost of a false negative and
$\lambda_0>0$ the cost of a false positive, and let
$\pi_k=P(D=k)$, $k\in\{0,1\}$, denote the class prevalences.
The expected weighted misclassification risk associated with
threshold $c$ is
\begin{align}
R(c;\theta)
&=
\lambda_1\,P(\text{FN at }c)
+
\lambda_0\,P(\text{FP at }c)
\nonumber\\
&=
\lambda_1\pi_1 F_1(c;\theta_1)
+
\lambda_0\pi_0\bigl[1-F_0(c;\theta_0)\bigr].
\label{eq:risk}
\end{align}
The first term in~\eqref{eq:risk} represents the weighted false
negative probability, while the second represents the weighted
false positive probability. Although the prevalence factors may
be absorbed into the loss coefficients by defining
$\tilde\lambda_k=\lambda_k\pi_k$, we retain costs and
prevalences separately throughout for interpretability.

Within this framework, threshold selection is naturally
formulated as a risk minimisation problem. That is, among all
admissible thresholds in $\mathcal C$, the optimal cutoff is
defined as any value that minimises the expected weighted
misclassification risk:
\begin{equation}
\label{eq:cutoff-def}
c^*(\theta)\in\arg\min_{c\in\mathcal C} R(c;\theta).
\end{equation}

To characterise the minimiser in~\eqref{eq:cutoff-def}, we
differentiate the risk function~\eqref{eq:risk} with respect to the threshold and obtain an equivalent first-order condition expressed in terms of the likelihood ratio.

\subsection{Likelihood-ratio characterisation}

Assume that $F_0(\cdot;\theta_0)$ and $F_1(\cdot;\theta_1)$ are
differentiable on $(a,b)$, with corresponding densities
$f_0(\cdot;\theta_0)$ and $f_1(\cdot;\theta_1)$.
Then the derivative of the risk with respect to the threshold is
\begin{equation}
\label{eq:estimating-eq}
\varphi(c,\theta)
:=
\partial_c R(c;\theta)
=
\lambda_1\pi_1 f_1(c;\theta_1)
-
\lambda_0\pi_0 f_0(c;\theta_0).
\end{equation}
Accordingly, any interior minimiser $c^*(\theta)\in(a,b)$ must
satisfy the first-order condition
\[
\varphi(c^*(\theta),\theta)=0.
\]
Equivalently,
\begin{equation}
\label{eq:likelihood-ratio-condition}
\Lambda(c^*(\theta);\theta)
:=
\frac{f_1(c^*(\theta);\theta_1)}{f_0(c^*(\theta);\theta_0)}
=
\frac{\lambda_0\pi_0}{\lambda_1\pi_1}.
\end{equation}
Thus, the optimal cutoff is characterised by the point at which
the likelihood ratio equals the cost-prevalence ratio.

By equation~\eqref{eq:roc-slope}, the likelihood ratio
$\Lambda(c;\theta)$ is precisely the slope of the ROC curve at
the operating point associated with threshold $c$. Hence,
equation~\eqref{eq:likelihood-ratio-condition} states that the
optimal cutoff corresponds to the point on the
$\mathrm{ROC}$ curve whose tangent slope equals
$(\lambda_0\pi_0)/(\lambda_1\pi_1)$
\citep{metz1978basic,pepe2003statistical}. In contrast, the
Youden cutoff $c_Y(\theta)$, defined in
Section~\ref{subsec:youden}, satisfies
\[
\Lambda(c_Y(\theta);\theta)=1,
\]
and therefore coincides with $c^*(\theta)$ if and only if $\lambda_0\pi_0=\lambda_1\pi_1.$
That is, the Youden rule is recovered exactly in the symmetric
decision setting.

\begin{remark}[Geometric interpretation and relation to AUC]
\label{rem:auc-connection}
\normalfont
Equation~\eqref{eq:likelihood-ratio-condition} gives a geometric
interpretation of the optimal cutoff in ROC space. Since
$\Lambda(c;\theta)$ is the tangent slope of the ROC curve at
threshold $c$, the condition $\Lambda(c^*(\theta);\theta)=
(\lambda_0\pi_0)/(\lambda_1\pi_1)$ means that the optimal cutoff is the operating point at which
the ROC tangent slope matches the cost-prevalence ratio.
Under the boundary conditions of
Proposition~\ref{prop:optimal-cutoff}, this target slope is
attained at some point along the curve.

This also highlights the distinction between the AUC and the
optimal cutoff: the AUC summarises global discrimination across
all thresholds, whereas the optimal cutoff is a local quantity
determined by the ROC geometry at the operating point relevant
to the decision problem.
\end{remark}

The geometric interpretation above is informative, but the
decision-theoretic framework also requires a precise analytical
treatment of the optimisation problem. We therefore turn to the
basic theoretical questions of existence and characterisation of
the minimiser of the risk function.

\subsection{Existence and first-order optimality}

We first establish that the optimisation problem is well defined.

\begin{proposition}[\textbf{Existence of the optimal cutoff}]
\label{prop:existence}
Assume that, for each fixed $\theta\in\Theta$, the map
$c\mapsto R(c;\theta)$ is continuous on $\mathcal C$.
Then the set of minimisers
\[
\mathcal M(\theta):=\arg\min_{c \, \in\,\mathcal C} R(c;\theta)
\]
is nonempty.
\end{proposition}

\begin{proof}
Fix $\theta\in\Theta$. Since $R(\cdot;\theta)$ is continuous on
the compact set $\mathcal C\subset\mathbb R$, the Weierstrass
extreme value theorem implies that $R(\cdot;\theta)$ attains
its minimum on $\mathcal C$. Therefore, there exists at least
one point $c^*(\theta)\in\mathcal C$ such that
\[
R(c^*(\theta);\theta)\le R(c;\theta),
\qquad \text{for all } c\in\mathcal C.
\]
Equivalently,
\[
c^*(\theta)\in\arg\min_{c\,\in\,\mathcal C}R(c;\theta)
=\mathcal M(\theta),
\]
and hence $\mathcal M(\theta)\neq\varnothing$.
\end{proof}

We next state the first-order necessary condition satisfied by
any interior global minimiser.

\begin{proposition}[\textbf{First-order necessary condition}]
\label{prop:first-order}
Assume that:
\begin{enumerate}
\item for each fixed $\theta\in\Theta$, the map
      $c\mapsto R(c;\theta)$ is differentiable on $(a,b)$, with
      derivative $\partial_c R(c;\theta)=\varphi(c,\theta)$ as
      defined in~\eqref{eq:estimating-eq};
\item $c^*(\theta)\in(a,b)$ is a global minimiser of
      $R(\cdot;\theta)$ over $\mathcal C=[a,b]$.
\end{enumerate}
Then
\[
\varphi(c^*(\theta),\theta)=0.
\]
Equivalently, $c^*(\theta)$ satisfies the
likelihood-ratio equation~\eqref{eq:likelihood-ratio-condition}.
\end{proposition}

\begin{proof}
Since $c^*(\theta)\in(a,b)$ is a global minimiser over $[a,b]$,
it is in particular a local minimiser of the differentiable map
$c\mapsto R(c;\theta)$ on the open interval $(a,b)$. Fermat's
theorem therefore gives
\[
\partial_c R(c^*(\theta);\theta)=\varphi(c^*(\theta),\theta)=0,
\]
which is equivalent to
\[
\frac{f_1(c^*(\theta);\theta_1)}{f_0(c^*(\theta);\theta_0)}
=
\frac{\lambda_0\pi_0}{\lambda_1\pi_1}.
\]
\end{proof}

\begin{remark}
\normalfont
Proposition~\ref{prop:first-order} provides only a necessary
condition for optimality: every interior global minimiser must
satisfy the equation $\varphi(c,\theta)=0$. By itself, however,
this condition does not guarantee either existence or uniqueness
of a minimiser. These properties are established in
Proposition~\ref{prop:optimal-cutoff} under a monotone
likelihood ratio assumption.
\end{remark}

\subsection{Existence, uniqueness and global optimality under MLR}

The previous proposition identifies a necessary first-order
condition for optimality, but does not ensure that the equation
$\varphi(c,\theta)=0$ has a unique solution or that such a
solution is globally optimal. To obtain these conclusions, we now impose a monotone likelihood ratio condition.

\subsubsection*{Local regularity and differentiability}

\begin{theorem}[\textbf{Local regularity of the optimal cutoff}]
\label{thm:regularity}
Let $c_0\in(a,b)$ and $\theta^*\in\Theta$ be fixed.
Suppose that:
\begin{enumerate}
    \item $\varphi(c_0,\theta^*)=0$;

    \item the map $(c,\theta)\mapsto\varphi(c,\theta)$ is
          continuously differentiable in a neighbourhood of
          $(c_0,\theta^*)$;

    \item $\partial_c\varphi(c_0,\theta^*)\neq 0$.
\end{enumerate}
Then there exist an open neighbourhood
$\mathcal U\subseteq\Theta$ of $\theta^*$ and a unique
continuously differentiable function
\[
c^*:\mathcal U\longrightarrow(a,b)
\]
such that
\[
\varphi\!\left(c^*(\theta),\theta\right)=0
\quad\text{for all }\theta\in\mathcal U,
\qquad
c^*(\theta^*)=c_0.
\]
Moreover, for each $\theta\in\mathcal U$,
\[
\nabla_\theta c^*(\theta)
=
-
\frac{\nabla_\theta \varphi\!\left(c^*(\theta),\theta\right)}
{\partial_c\varphi\!\left(c^*(\theta),\theta\right)}
\;\in\mathbb R^p.
\]
\end{theorem}

\begin{proof}
Assumptions~(1)--(3) allow us to apply the implicit function theorem
to the map
\[
\varphi:\mathcal V\subset \mathbb R\times\Theta\to\mathbb R,
\]
defined on an open neighbourhood $\mathcal V$ of $(c_0,\theta^*)$.
(see, e.g., \citealt[Theorem~9.28]{rudin1976principles}).
Therefore, there exist an open neighbourhood
$\mathcal U_0\subseteq\Theta$ of $\theta^*$, an open interval
$\mathcal W_0\subset \mathbb R$ containing $c_0$, and a unique
$C^1$ function
\[
\tilde c:\mathcal U_0\to\mathcal W_0
\]
such that
\[
\varphi(\tilde c(\theta),\theta)=0
\quad\text{for all }\theta\in\mathcal U_0,
\qquad
\tilde c(\theta^*)=c_0.
\]

\noindent
Since $c_0\in(a,b)$ and $(a,b)$ is an open interval, we may,
after possibly shrinking $\mathcal W_0$, assume that
\[
\mathcal W_0\subset(a,b).
\]
Define
\[
\mathcal U
:=
\{\theta\in\mathcal U_0:\tilde c(\theta)\in(a,b)\}.
\]
Because $\tilde c$ is continuous and $(a,b)$ is open, the set
$\mathcal U$ is an open neighbourhood of $\theta^*$.
Restricting $\tilde c$ to $\mathcal U$, we obtain a unique
$C^1$ function
\[
c^*:=\tilde c|_{\mathcal U}:\mathcal U\to(a,b)
\]
such that
\[
\varphi(c^*(\theta),\theta)=0
\quad\text{for all }\theta\in\mathcal U,
\qquad
c^*(\theta^*)=c_0.
\]

\noindent
Differentiating the identity
\[
\varphi(c^*(\theta),\theta)=0
\]
with respect to $\theta\in\Theta\subset\mathbb R^p$ and applying the
chain rule gives
\[
\partial_c\varphi(c^*(\theta),\theta)\,
\nabla_\theta c^*(\theta)
+
\nabla_\theta\varphi(c^*(\theta),\theta)
=
\mathbf 0
\in\mathbb R^p.
\]
By continuity of $\partial_c\varphi$ and the fact that
\[
\partial_c\varphi(c_0,\theta^*)\neq 0,
\]
after possibly shrinking $\mathcal U$ we may ensure that
\[
\partial_c\varphi(c^*(\theta),\theta)\neq 0
\qquad\text{for all }\theta\in\mathcal U.
\]
Therefore,
\[
\nabla_\theta c^*(\theta)
=
-
\frac{\nabla_\theta\varphi(c^*(\theta),\theta)}
{\partial_c\varphi(c^*(\theta),\theta)},
\]
which proves the result.
\end{proof}
\begin{remark}
\normalfont
Theorem~\ref{thm:regularity} is purely local in nature: it
asserts the existence of a unique continuously differentiable
branch of solutions to the equation $\varphi(c,\theta)=0$ in a
neighbourhood of $(c_0,\theta^*)$. By itself, it does not imply
that this branch corresponds to a global minimiser of the risk
function. In the present setting, that interpretation is
provided by Proposition~\ref{prop:optimal-cutoff}, which
establishes existence, uniqueness, and global optimality of the
cutoff under a monotone likelihood ratio condition. Thus, when
Proposition~\ref{prop:optimal-cutoff} applies, the point $c_0$
may be taken as the unique global minimiser at $\theta^*$, and
Theorem~\ref{thm:regularity} shows that this optimal cutoff
depends smoothly on the parameter.
\end{remark}

\subsubsection*{From local optimality to global structure}

The previous results show that any interior global minimiser
must satisfy $\varphi(c,\theta)=0$, and that this minimiser
is a locally smooth function of the parameter $\theta$.
However, neither result guarantees that the equation
$\varphi(c,\theta)=0$ has a unique solution, since
$\varphi$ may have multiple sign changes when the likelihood
ratio is not monotone.

A natural sufficient condition for global uniqueness is the
monotone likelihood ratio (MLR) property.
Under MLR, the map $c\mapsto\varphi(c,\theta)$ is
strictly monotone, so the first-order equation admits exactly
one solution, which must then be the global minimiser.
The following proposition makes this precise.

\begin{proposition}[\textbf{Existence, uniqueness and global optimality under MLR}]
\label{prop:optimal-cutoff}
Assume that:
\begin{enumerate}
    \item the densities $f_0(\cdot\,;\theta_0)$ and
          $f_1(\cdot\,;\theta_1)$ are continuous and strictly
          positive on $(a,b)$;

    \item the likelihood ratio
          \[
          \Lambda(c;\theta)
          =
          \frac{f_1(c;\theta_1)}{f_0(c;\theta_0)}
          \]
          is continuous and strictly increasing on $(a,b)$;

    \item the boundary conditions
          \[
          \lim_{c\,\downarrow\,a}\Lambda(c;\theta)
          <
          \frac{\lambda_0\pi_0}{\lambda_1\pi_1}
          \qquad\text{and}\qquad
          \lim_{c\,\uparrow\,b}\Lambda(c;\theta)
          >
          \frac{\lambda_0\pi_0}{\lambda_1\pi_1}
          \]
          hold.
\end{enumerate}
Then there exists a unique $c^*(\theta)\in(a,b)$ satisfying
\[
\Lambda(c^*(\theta);\theta)
=
\frac{\lambda_0\pi_0}{\lambda_1\pi_1}.
\]
Moreover, $c^*(\theta)$ is the unique global minimiser of
$R(c;\theta)$ over $[a,b]$.
\end{proposition}

\begin{proof}
\textit{Existence.}
Define
\[
h(c;\theta)
=
\frac{\lambda_1\pi_1}{\lambda_0\pi_0}\,\Lambda(c;\theta)-1,
\qquad c\in(a,b).
\]
Since $f_0(c;\theta_0)>0$ on $(a,b)$ by assumption~(1),
we may factor
\begin{equation}\label{eq:derivada}
\partial_c R(c;\theta)=\varphi(c,\theta)
=
\lambda_0\pi_0\,f_0(c;\theta_0)\,h(c;\theta),   
\end{equation}
so $\varphi$ and $h$ share the same sign on $(a,b)$.
The function $h(\cdot\,;\theta)$ is continuous by
assumption~(2), and the boundary conditions~(3) imply
\[
\lim_{c\,\downarrow\, a}h(c;\theta)<0,
\qquad
\lim_{c\,\uparrow\, b}h(c;\theta)>0.
\]
By the intermediate value theorem, there exists at least one
$c^*(\theta)\in(a,b)$ such that $h(c^*(\theta);\theta)=0$, equivalently
$\Lambda(c^*(\theta);\theta)=(\lambda_0\pi_0)/(\lambda_1\pi_1)$.

\medskip
\noindent
\textit{Uniqueness.}
Since $\Lambda(\cdot;\theta)$ is strictly increasing by
assumption~(2), so is $h(\cdot\,;\theta)$.
Therefore $h(\cdot;\theta)=0$ has at most one solution in
$(a,b)$, which establishes uniqueness.

\medskip
\noindent
\textit{Global optimality.}
Since $F_0(\cdot\,;\theta_0)$ and $F_1(\cdot\,;\theta_1)$ admit continuous
densities $f_0(\cdot\,;\theta_0)$ and $f_1(\cdot\,;\theta_1)$ on $(a,b)$,
the risk function $R(\cdot\,;\theta)$ is continuously differentiable on
$(a,b)$, with derivative as in \eqref{eq:derivada}. Since $h(\cdot;\theta)$ is strictly increasing and
$h(c^*(\theta);\theta)=0$, it follows that
\[
h(c;\theta)<0 \quad \text{for } c<c^*(\theta),
\qquad
h(c;\theta)>0 \quad \text{for } c>c^*(\theta).
\]
Because
\[
\partial_c R(c;\theta)
=
\lambda_0\pi_0 f_0(c;\theta_0)\,h(c;\theta),
\]
with $\lambda_0\pi_0>0$ and $f_0(c;\theta_0)>0$ on $(a,b)$,
we obtain
\[
\partial_c R(c;\theta)<0 \quad \text{for } c\in(a,c^*(\theta)),
\qquad
\partial_c R(c;\theta)>0 \quad \text{for } c\in(c^*(\theta),b).
\]

\noindent
Now let $a\le c_1<c_2\le c^*(\theta)$. Since $R(\cdot;\theta)$ is
continuous on $[c_1,c_2]$ and differentiable on $(c_1,c_2)$, the
fundamental theorem of calculus yields
\[
R(c_2;\theta)-R(c_1;\theta)
=
\int_{c_1}^{c_2}\partial_c R(t;\theta)\,dt.
\]
As $\partial_c R(t;\theta)<0$ for all $t\in(c_1,c_2)$, it follows that
\[
R(c_2;\theta)<R(c_1;\theta).
\]
Thus $R(\cdot;\theta)$ is strictly decreasing on $[a,c^*(\theta)]$. Similarly, for any $c^*(\theta)\le c_1<c_2\le b$,
\[
R(c_2;\theta)-R(c_1;\theta)
=
\int_{c_1}^{c_2}\partial_c R(t;\theta)\,dt>0,
\]
because $\partial_c R(t;\theta)>0$ for all $t\in(c_1,c_2)$.
Hence $R(\cdot;\theta)$ is strictly increasing on $[c^*(\theta),b]$.

\noindent
Therefore,
\[
R(c;\theta)>R(c^*(\theta);\theta)
\qquad\text{for all } c\in[a,b],\; c\neq c^*(\theta),
\]
so $c^*(\theta)$ is the unique global minimiser of $R(\cdot;\theta)$ on $[a,b]$.
\end{proof}

\begin{remark}[Plausibility of the MLR assumption in SMSN models]
\label{rem:mlr-sm-sn}
\normalfont
The monotone likelihood ratio assumption provides a convenient
sufficient condition for existence, uniqueness, and global
optimality of the cutoff. While this condition holds exactly in
classical location-shift models (for example, Gaussian models
with common variance), it need not hold for the full SMSN
family, where the two populations may differ in scale,
skewness, and tail behaviour.

Nevertheless, in many practical applications, including
serological studies, the two groups differ mainly in location,
with only moderate differences in dispersion and asymmetry.
Under such configurations, the likelihood ratio
$\Lambda(c;\theta)$ may be approximately monotone over the
region of overlap between the two class densities, which often
leads to a unique and well-defined solution of the optimality
equation in practice.

More generally, the MLR condition is sufficient but not
necessary for uniqueness. The optimal cutoff remains uniquely
defined whenever the equation $\varphi(c,\theta)=0$ has a
single solution, even if $\Lambda(c;\theta)$ is not globally
monotone. In practice, this may be assessed by inspecting the
sign-change structure of the map $c\mapsto\varphi(c,\theta)$.

Thus, global structural assumptions such as MLR provide a
convenient theoretical route to uniqueness, while the numerical
location and stability of the cutoff are governed by the local
behaviour of the likelihood ratio near the decision boundary.
\end{remark}

The optimality condition also admits a natural geometric
interpretation in ROC space, which clarifies the role of the
likelihood ratio in threshold selection.

\begin{remark}[ROC interpretation of the optimality condition]
\label{rem:roc-geometry}
\normalfont
The likelihood ratio $\Lambda(c;\theta)$ admits a natural
geometric interpretation in ROC space. Indeed, parametrising
the ROC curve by the threshold $c$ as
\[
c \mapsto \bigl(\mathrm{FPR}(c), \mathrm{TPR}(c)\bigr),
\]
one obtains, whenever the derivatives exist and
$f_0(c;\theta_0)>0$,
\[
\frac{d\,\mathrm{TPR}(c)}{d\,\mathrm{FPR}(c)}
=
\Lambda(c;\theta),
\]
so that the likelihood ratio coincides with the tangent slope
of the ROC curve at the operating point associated with
threshold $c$.

\noindent
Under this representation, the optimality condition
\[
\Lambda(c^*(\theta);\theta)
=
\frac{\lambda_0\pi_0}{\lambda_1\pi_1}
\]
states that the optimal cutoff is the point on the ROC curve
whose tangent slope equals the ratio of weighted costs and
prevalences. The boundary conditions ensure that the ROC slope
passes from values below to values above this target level,
thereby guaranteeing the existence of such a point. Under the
MLR assumption, this point is unique and corresponds to the
global optimum.

This geometric viewpoint links decision-theoretic optimality to
ROC analysis and clarifies how optimal thresholds arise as
tangency points determined by prevalence and cost asymmetries.
\end{remark}

\subsection{Relationship with the Youden cutoff}
\label{subsec:youden-relationship}

Having characterised $c^*(\theta)$ through the
likelihood-ratio condition~\eqref{eq:likelihood-ratio-condition},
we now study its departure from the Youden cutoff
$c_Y(\theta)$ introduced in Section~\ref{subsec:youden}.
The next result provides a first-order approximation to this
displacement in terms of the cost-prevalence asymmetry.

\begin{proposition}[\textbf{Relationship with the Youden cutoff}]
\label{prop:youden-relationship}

Let $c_Y(\theta)$ and $c^*(\theta)$ be the Youden and optimal
cutoffs defined in~\eqref{eq:youden-cutoff}
and~\eqref{eq:cutoff-def}, respectively.
Define the log-likelihood ratio
\[
g(c;\theta) = \log \Lambda(c;\theta)
= \log f_1(c;\theta_1) - \log f_0(c;\theta_0).
\]
Assume further that the MLR condition
of Proposition~\ref{prop:optimal-cutoff} holds, so that
$\Lambda(c;\theta)$ is strictly increasing in $c$, 
and that $g(c;\theta)$ is continuously differentiable in a
neighbourhood of $c_Y(\theta)$ with
$\partial_c g(c_Y(\theta);\theta) \neq 0$.

\noindent
Then:
\begin{enumerate}
  \item If $\lambda_0\pi_0 = \lambda_1\pi_1$, then
        $c^*(\theta) = c_Y(\theta)$.

  \item Otherwise, as
        $\log(\lambda_0\pi_0/\lambda_1\pi_1) \to 0$,
        \begin{equation}
        \label{eq:youden-expansion}
        c^*(\theta) - c_Y(\theta)
        =
        \frac{
          \log\!\left(\dfrac{\lambda_0\pi_0}{\lambda_1\pi_1}\right)
        }{
          \partial_c g(c_Y(\theta);\theta)
        }
        +
        o\!\left(
          \log \left(\!\frac{\lambda_0\pi_0}{\lambda_1\pi_1}\right)
        \right),
        \end{equation}
        where
        \begin{equation}
\label{eq:dcg-formula}
\partial_c g(c_Y(\theta);\theta)
= \frac{
    \partial_c f_1(c_Y(\theta);\theta_1)
    - \partial_c f_0(c_Y(\theta);\theta_0)
  }{
    f_1(c_Y(\theta);\theta_1)
  },
\end{equation}
where $f_1(c_Y;\theta_1) = f_0(c_Y;\theta_0)$ implies that the
denominator may be equivalently written as $f_0(c_Y(\theta);\theta_0)$.
\end{enumerate}
\end{proposition}

\begin{proof}
By condition~\eqref{eq:youden-condition}, $g(c_Y(\theta);\theta) = \log (1) = 0$. 
By condition~\eqref{eq:likelihood-ratio-condition},
$g(c^*(\theta);\theta) = \log(\lambda_0\pi_0/\lambda_1\pi_1)$.
A first-order Taylor expansion of $g(\cdot;\theta)$
around $c_Y(\theta)$ gives
\[
\log \left(\!\frac{\lambda_0\pi_0}{\lambda_1\pi_1}\right)
= \partial_c g(c_Y(\theta);\theta)
  \bigl(c^*(\theta) - c_Y(\theta)\bigr)
  + o\bigl(c^*(\theta) - c_Y(\theta)\bigr).
\]
Since $\partial_c g(c_Y(\theta);\theta) \neq 0$ by assumption,
solving for $c^*(\theta) - c_Y(\theta)$ yields~\eqref{eq:youden-expansion}.
Formula~\eqref{eq:dcg-formula} follows by differentiating
$g(c;\theta)$ and evaluating at $c = c_Y(\theta)$, where
$f_1(c_Y;\theta_1) = f_0(c_Y;\theta_0)$.
\end{proof}

\begin{remark}[Interpretation and local sensitivity]
\label{rem:youden-interpretation}
\normalfont
The expansion~\eqref{eq:youden-expansion} shows that the
displacement $\Delta_c(\theta)=c^*(\theta)-c_Y(\theta)$ is
governed by two components. The numerator,
$\log(\lambda_0\pi_0/\lambda_1\pi_1)$, measures the degree of
asymmetry in the decision problem: it vanishes under joint
symmetry of costs and prevalences, and its sign determines the
direction of the shift from the Youden cutoff.

The denominator, $\partial_c g(c_Y(\theta);\theta)$, measures
the local steepness of the log-likelihood ratio at the Youden
point. When this quantity is large in magnitude, the
likelihood ratio changes rapidly near $c_Y(\theta)$, so the
optimal cutoff remains close to the Youden cutoff even under
moderate asymmetry. Conversely, when it is small in magnitude,
the likelihood ratio is locally flat, and small departures from
the symmetric setting may produce substantial shifts in the
optimal threshold.
\end{remark}

The previous results characterise the optimal cutoff at the
population level, assuming that the model parameter $\theta$ is
known. In practice, however, $\theta$ must be estimated from the
data. This naturally leads to a plug-in estimator of the
optimal cutoff obtained by replacing $\theta$ with a suitable
estimator $\hat\theta$ in the defining optimality equation.

\subsection{Plug-in estimation of the optimal cutoff}

In practice, the parameter vector $\theta^*$ is unknown and must
be estimated from the data. Let
\[
\hat\theta=(\hat\theta_0^\top,\hat\theta_1^\top)^\top
\]
denote an estimator of $\theta^*$, for example the maximum
likelihood estimator under the SMSN model.

Motivated by the likelihood-ratio characterisation of the
optimal cutoff, we define the plug-in estimator $\hat c$ as a
solution of the sample analogue of the estimating equation,
\begin{equation*}
\label{eq:plugin}
\varphi(\hat c,\hat\theta)=0,
\end{equation*}
that is,
\begin{equation}
\label{eq:estimated-eq}
\lambda_1\pi_1 f_1(\hat c;\hat\theta_1)
-
\lambda_0\pi_0 f_0(\hat c;\hat\theta_0)
=0.
\end{equation}

\noindent
At the population level, the optimal cutoff
$c^*=c^*(\theta^*)$ satisfies
\[
\varphi(c^*,\theta^*)=0.
\]
Assume further that this root is non-degenerate, in the sense
that
\[
\partial_c\varphi(c^*,\theta^*)\neq 0.
\]
Then, by Theorem~\ref{thm:regularity}, there exists an open
neighbourhood $\mathcal U$ of $\theta^*$ and a unique
continuously differentiable function
\[
\theta\mapsto c^*(\theta)
\]
defined on $\mathcal U$ such that
\[
\varphi(c^*(\theta),\theta)=0
\qquad\text{for all }\theta\in\mathcal U,
\qquad
c^*(\theta^*)=c^*.
\]

\noindent
Consequently, whenever $\hat\theta\in\mathcal U$, the plug-in
estimator may be written equivalently as
\[
\hat c=c^*(\hat\theta).
\]
This representation is fundamental for asymptotic analysis,
since it allows the large-sample behaviour of $\hat c$ to be
derived from that of $\hat\theta$ through the implicit function
theorem and the multivariate delta method.

In practice, equation~\eqref{eq:estimated-eq} will generally be
solved numerically. The regularity assumptions introduced below
ensure that the resulting estimator is well defined locally and
inherits the smooth structure of the population cutoff.

To study the large-sample behaviour of $\hat c$, we now
formalise the assumptions required on the parametric model, the
estimating equation, and the estimator $\hat\theta$.

\subsection{Regularity conditions for asymptotic inference}

We now study the asymptotic behaviour of the plug-in estimator
$\hat c$. Suppose that two independent samples are observed,
\[
X_{01},\ldots,X_{0n_0}\stackrel{\mathrm{i.i.d.}}{\sim}
f_0(\cdot;\theta_0^*),
\qquad
X_{11},\ldots,X_{1n_1}\stackrel{\mathrm{i.i.d.}}{\sim}
f_1(\cdot;\theta_1^*),
\]
where
\[
\theta^*=(\theta_0^{*\top},\theta_1^{*\top})^\top,
\qquad
n=n_0+n_1\to\infty.
\]

\noindent
We impose the following regularity conditions.

\begin{itemize}
\item[(i)] For each $k\in\{0,1\}$, the parameter space
$\Theta_k\subset\mathbb R^{p_k}$ is open, and
$\theta_k^*$ lies in its interior.

\item[(ii)] For each $k\in\{0,1\}$, the model is identifiable, in the sense
that
\[
f_k(\cdot;\theta_k)=f_k(\cdot;\theta_k')
\quad\text{almost everywhere}
\quad\Longrightarrow\quad
\theta_k=\theta_k'.
\]
This assumption is satisfied for the standard univariate SMSN
submodels considered in this paper, in particular, the
skew-normal and skew-$t$ distributions, provided that the true parameter lies in the interior of the parameter space and within a regular region where the model is identifiable. In this framework, the score vector and the Hessian matrix are well-defined, and the Fisher information matrix is non-singular. In particular, we exclude degenerate or near-boundary configurations, such as vanishing scale and, for the skew-$t$ model, values of $\nu$ arbitrarily close to $2$.

\item[(iii)] For each $k\in\{0,1\}$ and each fixed $x$, the map $\theta_k\mapsto f_k(x;\theta_k)$ is twice continuously differentiable in a neighbourhood of $\theta_k^*$. Moreover, we assume the usual regularity conditions under which derivatives with respect to the parameter may be taken under the integral sign, the score has mean zero, the Fisher information is well-defined and non-singular, and the standard asymptotic theory for the maximum likelihood estimators applies.

\item[(iv)] The estimator $\hat\theta$ satisfies
\begin{equation}
\label{eq:theta-an}
\sqrt{n}\bigl(\hat\theta-\theta^*\bigr)
\stackrel{d}{\longrightarrow}
\mathcal N(0,\Sigma),
\end{equation}
for some positive definite matrix $\Sigma$.
In the present setting, this condition holds for the maximum likelihood estimator under the usual regularity assumptions for regular univariate skew-normal and skew-$t$ models, namely: identifiability, interior true parameter, sufficient smoothness of the log-likelihood, validity of differentiation under the integral sign, and non-singularity of the Fisher information; see, e.g., \citet{arevalillo2012maximum},
\citet{azzalini2014skew}, and \citet{van2000asymptotic}.

\item[(v)] The population optimal cutoff is non-degenerate, that is,
\[
\partial_c\varphi(c^*,\theta^*)\neq 0.
\]
This condition ensures that the solution of the estimating equation $\varphi(c,\theta)=0$ is locally unique and continuously differentiable as a function of $\theta$. Throughout the paper, we assume that the fitted SMSN submodels satisfy conditions~(i)--(v). Under these assumptions, the
plug-in estimator $\hat c$ inherits its large-sample behaviour from that of $\hat\theta$ through the implicit function theorem and the multivariate delta method.
\end{itemize}

\subsection{Consistency and asymptotic normality of the cutoff estimator}

We now formalise the large-sample properties of $\hat c$, building on the implicit function representation of the optimal cutoff. For all sufficiently large $n$, let $\hat c$ denote the unique solution of $\varphi(\hat c,\hat\theta)=0$ in a neighbourhood of
$c^*$.

\begin{remark}[Stochastic order notation]
\label{rem:stochastic-order}
Throughout the asymptotic arguments below, we use the standard
stochastic order notation. Let $(a_n)$ be a sequence of positive real numbers.
\begin{itemize}
  \item We write $X_n = O_p(a_n)$ if the sequence $X_n/a_n$ is bounded
        in probability: for every $\varepsilon > 0$ there exist constants
        $M_\varepsilon < \infty$ and $N_\varepsilon$ such that
        \[
        P\!\left(\left|\frac{X_n}{a_n}\right| > M_\varepsilon\right)
        < \varepsilon
        \quad \text{for all } n \geq N_\varepsilon.
        \]
  \item We write $X_n = o_p(a_n)$ if $X_n/a_n \xrightarrow{P} 0$.
        In particular, $X_n = o_p(a_n)$ implies $X_n = O_p(a_n)$.
\end{itemize}
\end{remark}

\begin{theorem}[\textbf{Consistency and asymptotic normality}]
\label{thm:asymptotics}
Assume the conditions of Proposition~\ref{prop:optimal-cutoff}. Suppose further that:
\begin{enumerate}
\item the map $(c,\theta)\mapsto\varphi(c,\theta)$ is continuously differentiable in a neighbourhood of $(c^*,\theta^*)$, where $c^*=c^*(\theta^*)$;
\item
\[
\partial_c\varphi(c^*,\theta^*)\neq 0;
\]
\item the estimator $\hat\theta$ satisfies
\[
\sqrt{n}(\hat\theta-\theta^*)\stackrel{d}{\longrightarrow}\mathcal N(0,\Sigma).
\]
\end{enumerate}
Then:
\begin{enumerate}
\item $\hat c\xrightarrow{P}c^*$;
\item
\[
\sqrt{n}(\hat c-c^*)\stackrel{d}{\longrightarrow}\mathcal N(0,V),
\]
where
\begin{equation}
\label{eq:V-general}
V
=
\nabla_\theta c^*(\theta^*)^\top\Sigma\,\nabla_\theta c^*(\theta^*),
\end{equation}
and
\begin{equation}
\label{eq:gradient-cstar}
\nabla_\theta c^*(\theta^*)
= - \frac{\nabla_\theta\varphi(c^*,\theta^*)}{\partial_c\varphi(c^*,\theta^*)}.
\end{equation}
Equivalently,
\begin{equation}
\label{eq:variance}
V
=
\frac{
\bigl[\nabla_\theta\varphi(c^*,\theta^*)\bigr]^\top
\Sigma
\bigl[\nabla_\theta\varphi(c^*,\theta^*)\bigr]
}{
\bigl[\partial_c\varphi(c^*,\theta^*)\bigr]^2
}.
\end{equation}
\end{enumerate}
\end{theorem}

\begin{proof}
Since $\varphi(c^*,\theta^*)=0$ and $\partial_c\varphi(c^*,\theta^*)\neq 0$, the implicit function theorem guarantees the existence of a neighbourhood $\mathcal U$ of $\theta^*$ and a unique continuously differentiable function
\[
\theta\mapsto c^*(\theta)
\]
such that
\[
\varphi(c^*(\theta),\theta)=0,
\qquad
\theta\in\mathcal U,
\]
with $c^*(\theta^*)=c^*$.

\noindent
Because $\hat\theta\xrightarrow{P}\theta^*$ and $c^*(\cdot)$ is continuous at $\theta^*$, the continuous mapping theorem yields
\[
\hat c=c^*(\hat\theta)\xrightarrow{P}c^*(\theta^*)=c^*.
\]
This proves consistency.

For asymptotic normality, apply a first-order Taylor expansion of the differentiable map $\theta \mapsto c^*(\theta)$ around $\theta^*$:
\[
c^*(\hat\theta)-c^*(\theta^*)
=
\nabla_\theta c^*(\theta^*)^\top(\hat\theta-\theta^*)
+
\ell (\hat\theta),
\]
where the remainder term $\ell(\hat\theta)$ satisfies
\[
\frac{\ell(\hat\theta)}{\|\hat\theta-\theta^*\|} \xrightarrow{P} 0.
\]

\noindent
Since $\hat\theta$ is $\sqrt{n}$-consistent, i.e.,
\[
\hat\theta - \theta^* = O_p(n^{-1/2}),
\]
it follows that
\[
\ell(\hat\theta)
=
\|\hat\theta-\theta^*\|
\cdot
\frac{\ell(\hat\theta)}{\|\hat\theta-\theta^*\|}
=
O_p(n^{-1/2}) \cdot o_p(1)
=
o_p(n^{-1/2}).
\]

\noindent
Therefore,
\[
c^*(\hat\theta)-c^*(\theta^*)
=
\nabla_\theta c^*(\theta^*)^\top(\hat\theta-\theta^*)
+
o_p(n^{-1/2}).
\]

\noindent
Multiplying both sides by $\sqrt{n}$ yields
\[
\sqrt{n}(\hat c-c^*)
=
\nabla_\theta c^*(\theta^*)^\top \sqrt{n}(\hat\theta-\theta^*)
+
o_p(1).
\]
By~\eqref{eq:theta-an} and the multivariate delta method,
\[
\sqrt n(\hat c-c^*)
\stackrel{d}{\longrightarrow}
\mathcal N\!\left(
0,
\nabla_\theta c^*(\theta^*)^\top\Sigma\,\nabla_\theta c^*(\theta^*)
\right),
\]
which proves~\eqref{eq:V-general}.

\noindent
It remains to compute $\nabla_\theta c^*(\theta^*)$. Differentiate the identity
\[
\varphi(c^*(\theta),\theta)=0
\]
with respect to $\theta$. By the chain rule,
\[
\partial_c\varphi(c^*(\theta),\theta)\,\nabla_\theta c^*(\theta)
+
\nabla_\theta\varphi(c^*(\theta),\theta)
=
0.
\]
Evaluating at $\theta=\theta^*$ and solving for $\nabla_\theta c^*(\theta^*)$ yields
\[
\nabla_\theta c^*(\theta^*)
= -\frac{\nabla_\theta\varphi(c^*,\theta^*)}{\partial_c\varphi(c^*,\theta^*)},\]
which proves~\eqref{eq:gradient-cstar}. Substituting this into~\eqref{eq:V-general} gives~\eqref{eq:variance}. The explicit expressions follow from
\[
\partial_c\varphi(c^*,\theta^*)
=
\lambda_1\pi_1 \partial_c f_1(c^*;\theta_1^*)-\lambda_0\pi_0 \partial_c f_0(c^*;\theta_0^*)
\]
and
\[
\nabla_\theta\varphi(c^*,\theta^*)
=
\begin{pmatrix}
-\lambda_0\pi_0\,\nabla_{\theta_0}f_0(c^*;\theta_0^*)\\[0.4em]
\lambda_1\pi_1\,\nabla_{\theta_1}f_1(c^*;\theta_1^*)
\end{pmatrix}.
\]
\end{proof}

\begin{remark}[Interpretation of the asymptotic variance]
\label{rem:variance-interpretation} 

The asymptotic variance
\[
V
=
\frac{
\bigl[\nabla_\theta \varphi(c^*,\theta^*)\bigr]^\top
\Sigma
\bigl[\nabla_\theta \varphi(c^*,\theta^*)\bigr]
}{
\bigl[\partial_c \varphi(c^*,\theta^*)\bigr]^2
}
\]
admits a natural interpretation in terms of sensitivity and local information. The numerator reflects the propagation of uncertainty from the parameter estimator $\hat\theta$ to the cutoff through the gradient $\nabla_\theta \varphi(c^*,\theta^*)$. It combines the sensitivity of the decision rule to the model parameters with their estimation variability, as encoded in the covariance matrix $\Sigma$. The denominator, $\partial_c \varphi(c^*,\theta^*)$, captures the local slope of the estimating equation at the optimal cutoff. This quantity measures how sharply the decision boundary is defined. When its magnitude is large, the equation $\varphi(c,\theta)=0$ crosses zero steeply, leading to a stable and well-identified cutoff. Conversely, when $\partial_c \varphi(c^*,\theta^*)$ is close to zero, the boundary becomes locally flat, and small perturbations in $\theta$ may induce large changes in the solution $c^*(\theta)$.

In this sense, $\bigl|\partial_c\varphi(c^*,\theta^*)\bigr|$
plays a role analogous to that of information in classical
parametric inference: just as low Fisher information leads to high variance of estimators, a small slope of the estimating equation leads to inflated asymptotic variance and weak identifiability of the cutoff.
Overall, the variability of the estimated cutoff is jointly governed by uncertainty in parameter estimation and the local geometry of the decision boundary, highlighting the importance of both model fit and discriminative structure.
\end{remark}

We now detail the general asymptotic result of Theorem~\ref{thm:asymptotics} to the practically relevant setting in which the two populations are modelled separately using SMSN distributions.

\begin{proposition}[\textbf{Asymptotic distribution under separate SMSN maximum likelihood fits}]
\label{prop:separate-mle}

Assume that the two groups are independent, with
\[
X_{01},\ldots,X_{0n_0}\stackrel{\mathrm{i.i.d.}}{\sim} f_0(\cdot;\theta_0^*),
\qquad
X_{11},\ldots,X_{1n_1}\stackrel{\mathrm{i.i.d.}}{\sim} f_1(\cdot;\theta_1^*),
\]
where, for each $k\in\{0,1\}$, $f_k(\cdot;\theta_k)$ belongs to a regular
univariate SMSN subfamily, in particular the skew-normal or skew-$t$
family, with true parameter
$\theta_k^*\in\Theta_k\subset\mathbb R^{p_k}$.
Suppose that $\hat\theta_0$ and $\hat\theta_1$ are the maximum
likelihood estimators obtained by fitting the two models separately to
groups $0$ and $1$, respectively, and let
\[
\hat\theta=(\hat\theta_0^\top,\hat\theta_1^\top)^\top,
\qquad
\theta^*=(\theta_0^{*\top},\theta_1^{*\top})^\top.
\]

\noindent
Assume further that:
\begin{enumerate}
\item the conditions of Proposition~\ref{prop:optimal-cutoff} hold, so that the optimal cutoff $c^*=c^*(\theta^*)$ is uniquely defined;
\item the map $(c,\theta)\mapsto\varphi(c,\theta)$ is continuously differentiable in a neighbourhood of $(c^*,\theta^*)$, and
\[
\partial_c\varphi(c^*,\theta^*)\neq 0;
\]
\item for each $k\in\{0,1\}$, the groupwise SMSN log-likelihood satisfies the usual regularity conditions for maximum likelihood estimation, with per-observation Fisher information matrix $I_k(\theta_k^*)$ finite and positive definite;
\item the sample proportions satisfy
\[
\frac{n_0}{n}\to\eta_0\in(0,1),
\quad
\frac{n_1}{n}\to\eta_1\in(0,1),
\quad
n=n_0+n_1\to\infty.
\]
\end{enumerate}
Then,
\begin{enumerate}
\item the joint estimator $\hat\theta$ is consistent for $\theta^*$ and satisfies
\begin{equation}
\label{eq:joint-an-separate}
\sqrt{n}\,(\hat\theta-\theta^*)
\stackrel{d}{\longrightarrow}
\mathcal N(0,\Sigma_{\mathrm{sep}}),
\end{equation}
where
\begin{equation}
\label{eq:sigma-sep}
\Sigma_{\mathrm{sep}}
=
\begin{pmatrix}
\eta_0^{-1}I_0(\theta_0^*)^{-1} & 0\\[0.4em]
0 & \eta_1^{-1}I_1(\theta_1^*)^{-1}
\end{pmatrix};
\end{equation}

\item let $c^*(\theta)$ denote the local branch of solutions to
$\varphi(c,\theta)=0$ defined in a neighbourhood of $\theta^*$ by
Theorem~\ref{thm:regularity}, and define the plug-in estimator by
$\hat c=c^*(\hat\theta)$.
Then $\hat c$ is consistent for $c^*$ and
\begin{equation}
\label{eq:chat-an-separate}
\sqrt{n}\,(\hat c-c^*)
\stackrel{d}{\longrightarrow}
\mathcal N(0,V_{\mathrm{sep}}),
\end{equation}
with asymptotic variance
\begin{equation}
\label{eq:v-sep}
V_{\mathrm{sep}}
=
\frac{
\bigl[\nabla_\theta\varphi(c^*,\theta^*)\bigr]^\top
\Sigma_{\mathrm{sep}}
\bigl[\nabla_\theta\varphi(c^*,\theta^*)\bigr]
}{
\bigl[\partial_c\varphi(c^*,\theta^*)\bigr]^2
}.
\end{equation}
Equivalently,
\begin{equation}
\label{eq:v-sep-expanded}
\begin{split}
V_{\mathrm{sep}}
&=
\frac{
\lambda_0^2\pi_0^2\,
\nabla_{\theta_0}f_0(c^*;\theta_0^*)^\top
\eta_0^{-1}I_0(\theta_0^*)^{-1}
\nabla_{\theta_0}f_0(c^*;\theta_0^*)
}{
\bigl[
\lambda_1\pi_1 \partial_c f_1(c^*;\theta_1^*)
-
\lambda_0\pi_0 \partial_c f_0(c^*;\theta_0^*)
\bigr]^2
}
\\
&\quad+
\frac{
\lambda_1^2\pi_1^2\,
\nabla_{\theta_1}f_1(c^*;\theta_1^*)^\top
\eta_1^{-1}I_1(\theta_1^*)^{-1}
\nabla_{\theta_1}f_1(c^*;\theta_1^*)
}{
\bigl[
\lambda_1\pi_1 \partial_c f_1(c^*;\theta_1^*)
-
\lambda_0\pi_0 \partial_c f_0(c^*;\theta_0^*)
\bigr]^2
}.
\end{split}
\end{equation}
\end{enumerate}
\end{proposition}

\begin{proof}
Because the seronegative and seropositive samples are independent and the two SMSN models are fitted separately, the estimators $\hat\theta_0$ and $\hat\theta_1$ are asymptotically independent. For each $k\in\{0,1\}$, standard likelihood theory for regular parametric models yields
\[
\sqrt{n_k}\,(\hat\theta_k-\theta_k^*)
\stackrel{d}{\longrightarrow}
\mathcal N\bigl(0,I_k(\theta_k^*)^{-1}\bigr).
\]
Since $n_k/n\to\eta_k$, we may rewrite this in the $\sqrt n$-scale as
\[
\sqrt n\,(\hat\theta_k-\theta_k^*)
=
\sqrt{\frac{n}{n_k}}
\sqrt{n_k}\,(\hat\theta_k-\theta_k^*)
\stackrel{d}{\longrightarrow}
\mathcal N\bigl(0,\eta_k^{-1}I_k(\theta_k^*)^{-1}\bigr).
\]
By independence between groups, it follows that
\[
\sqrt n\,(\hat\theta-\theta^*)
\stackrel{d}{\longrightarrow}
\mathcal N(0,\Sigma_{\mathrm{sep}}),
\]
with $\Sigma_{\mathrm{sep}}$ given by~\eqref{eq:sigma-sep}. This proves~\eqref{eq:joint-an-separate}.

The asymptotic distribution of $\hat c=c^*(\hat\theta)$ then follows
immediately from Theorem~\ref{thm:asymptotics}, applied with
$\Sigma=\Sigma_{\mathrm{sep}}$. Indeed,
\[
\sqrt n\,(\hat c-c^*)
\stackrel{d}{\longrightarrow}
\mathcal N\!\left(
0,\,
\frac{
\bigl[\nabla_\theta\varphi(c^*,\theta^*)\bigr]^\top
\Sigma_{\mathrm{sep}}
\bigl[\nabla_\theta\varphi(c^*,\theta^*)\bigr]
}{
\bigl[\partial_c\varphi(c^*,\theta^*)\bigr]^2
}
\right),
\]
which gives~\eqref{eq:v-sep}. Finally, using the block structure of $\Sigma_{\mathrm{sep}}$ together with
\[
\nabla_\theta\varphi(c^*,\theta^*)
=
\begin{pmatrix}
-\lambda_0\pi_0\,\nabla_{\theta_0}f_0(c^*;\theta_0^*)\\[0.4em]
\lambda_1\pi_1\,\nabla_{\theta_1}f_1(c^*;\theta_1^*)
\end{pmatrix},
\]
and
\[
\partial_c\varphi(c^*,\theta^*)
=
\lambda_1\pi_1 \partial_c f_1(c^*;\theta_1^*)
-
\lambda_0\pi_0 \partial_c f_0(c^*;\theta_0^*),
\]
we obtain the expanded formula~\eqref{eq:v-sep-expanded}.
\end{proof}

The asymptotic results above provide the theoretical basis for statistical inference on the optimal cutoff. In practice, however, the asymptotic variance depends on unknown population quantities and must therefore be estimated. We now introduce a plug-in estimator of this variance and the corresponding Wald-type confidence interval for the cutoff.

\subsection{Estimated asymptotic variance and confidence interval}

In practice, the asymptotic variance $V$ is unknown because it depends on the unknown quantities $c^*$, $\theta^*$, and $\Sigma$. A natural estimator is obtained by plug-in substitution:
\begin{equation}
\label{eq:V-hat}
\hat V
=
\frac{
\bigl[\nabla_\theta\varphi(\hat c,\hat\theta)\bigr]^\top
\hat\Sigma
\bigl[\nabla_\theta\varphi(\hat c,\hat\theta)\bigr]
}{
\bigl[\partial_c\varphi(\hat c,\hat\theta)\bigr]^2
},
\end{equation}
where $\hat\Sigma$ is a consistent estimator of $\Sigma$. In the present setting, $\hat\Sigma$ is obtained from the block-diagonal inverse observed information matrices associated with the two separately fitted group models.

\noindent
Accordingly, the asymptotic standard error of $\hat c$ is estimated by
\[
\operatorname{se}(\hat c)=\sqrt{\hat V/n},
\]
and an approximate $(1-\alpha)$ Wald confidence interval for the optimal cutoff is given by
\begin{equation}
\label{eq:ci}
\hat c\pm z_{1-\alpha/2}\sqrt{\hat V/n},
\end{equation}
where $z_{1-\alpha/2}$ is the $(1-\alpha/2)$ quantile of the standard normal distribution.

Although this interval is asymptotically justified, its finite-sample performance may deteriorate when the estimating equation is locally flat around the optimum, when the fitted SMSN model is numerically unstable, or when the sampling distribution of $\hat c$ is noticeably skewed. In such cases, bootstrap resampling within each group may be used as a complementary inferential tool.

\subsubsection*{From theory to empirical behaviour}

The theoretical results developed above show that the optimal cutoff selection depends not only on the global separation between the two populations, but also on the local geometry of the decision boundary. In particular, both the displacement from the Youden cutoff and the asymptotic variance of the plug-in estimator are governed by quantities
that reflect the local behaviour of the estimating equation, or equivalently of the likelihood ratio, near the optimal threshold.

This has important practical consequences. A biomarker may display satisfactory global discrimination, yet still yield a weakly identifiable cutoff if the decision boundary is locally flat in the overlap region between the two class densities. In such cases, the optimal threshold becomes more sensitive to model perturbations and more variable in finite samples. Conversely, when the likelihood ratio changes steeply near the decision boundary, the cutoff is more stable and more precisely estimated.

These considerations motivate an empirical assessment of the proposed methodology. In the following sections, we examine these behaviours through Monte Carlo simulations and their application to serological data. The simulation study is designed to assess the finite-sample behaviour of the plug-in estimator, the accuracy of the asymptotic variance approximation, and the coverage of Wald confidence intervals under different distributional configurations. The real-data analysis then illustrates how these theoretical insights translate into practical
decision-making in a biomedical setting.

\section{Simulation Study}
\label{sec:simulation}

A simulation study was conducted with two complementary objectives:
(i)~assess empirically the asymptotic theory established in
Theorem~\ref{thm:asymptotics}, namely the consistency of the plug-in
estimator $\hat{c}$, the adequacy of the normal approximation to the
distribution of $\sqrt{n}(\hat{c} - c^*)$, and the accuracy of the
plug-in variance estimator $\hat{V}$; (ii)~characterise the
finite-sample behaviour of the estimator across a range of practically
relevant distributional regimes, with particular emphasis on the role of
the local identifiability diagnostic $|\partial_c\varphi(c^*,\theta^*)|$
as a measure of the stability of the decision boundary.

\subsection{Scenario design}

Six simulation scenarios were constructed to span a
representative range of configurations of the decision
problem, combining different distributional families,
degrees of class separation, and cost--prevalence
structures. A preliminary numerical screening was conducted for each scenario to verify that the estimating equation $c\mapsto\varphi(c,\theta)$ admits a unique non-degenerate root over the effective support of the data, as required by the regularity conditions underlying Theorem~\ref{thm:asymptotics}.
Only configurations satisfying this criterion were retained.

Three scenarios were specified under the skew-normal family and three under the skew-$t$ family, as detailed in Tables~\ref{tab:scenario-metadata-sn}
and~\ref{tab:scenario-metadata-st}.

\begin{table}[ht]
\centering
\caption{Simulation scenarios under the skew-normal distribution.}
\label{tab:scenario-metadata-sn}
\small
\setlength{\tabcolsep}{4pt}
\begin{tabular}{@{}llcccccccccc@{}}
\toprule
\textbf{Scenario} & \textbf{Description} & $\lambda_0$ & $\lambda_1$ &
  $\pi_0$ & $\pi_1$ & $\xi_0$ & $\omega_0$ & $\alpha_0$ &
  $\xi_1$ & $\omega_1$ & $\alpha_1$ \\
\midrule
SN1 & Balanced, moderate separation & 1 & 1 & 0.5 & 0.5 & 0 & 1 & 1 & 2 & 1 & 1.5 \\
SN2 & Asymmetric prevalence/cost   & 1 & 3 & 0.8 & 0.2 & 0 & 1 & 1.5 & 2 & 1.2 & 2 \\
SN3 & Weak separation              & 1 & 1 & 0.5 & 0.5 & 0 & 1 & 1 & 1 & 1 & 1.2 \\
\bottomrule
\end{tabular}
\medskip\\
{\footnotesize\textit{Note:} $(\lambda_0,\lambda_1)$ denote misclassification costs;
$(\pi_0,\pi_1)$ denote class prevalences; $(\xi,\omega,\alpha)$ are
location, scale, and skewness of the skew-normal distribution.}
\end{table}

\begin{table}[ht]
\centering
\caption{Simulation scenarios under the skew-$t$ distribution.}
\label{tab:scenario-metadata-st}
\small
\setlength{\tabcolsep}{3pt}
\begin{tabular}{@{}llcccccccccccc@{}}
\toprule
\textbf{Scenario} & \textbf{Description} & $\lambda_0$ & $\lambda_1$ &
  $\pi_0$ & $\pi_1$ & $\xi_0$ & $\omega_0$ & $\alpha_0$ & $\nu_0$ &
  $\xi_1$ & $\omega_1$ & $\alpha_1$ & $\nu_1$ \\
\midrule
ST1 & Balanced, moderate tails      & 1 & 1 & 0.5 & 0.5 & 0 & 1 & 1 & 8 & 2 & 1 & 1.5 & 8 \\
ST2 & Asymmetric prevalence/cost    & 1 & 3 & 0.8 & 0.2 & 0 & 1 & 1.5 & 7 & 2 & 1.2 & 2 & 7 \\
ST3 & Weak separation, heavy tails  & 1 & 1 & 0.5 & 0.5 & 0 & 1 & 1 & 5 & 1 & 1 & 1.2 & 5 \\
\bottomrule
\end{tabular}
\medskip\\
{\footnotesize\textit{Note:} In addition to location, scale, and skewness parameters,
$\nu$ denotes the degrees of freedom controlling tail heaviness.}
\end{table}

\noindent
Within each family, the scenarios were designed to represent qualitatively distinct regimes:
\begin{itemize}
  \item \textup{(SN1, ST1):} a balanced symmetric setting with moderate class separation, serving as a baseline reference;
  \item \textup{(SN2, ST2):} an asymmetric setting with
        unequal costs and prevalences
        $(\lambda_0=1,\lambda_1=3;\,\pi_0=0.8,\pi_1=0.2)$,
        so that $\lambda_0\pi_0\neq\lambda_1\pi_1$ and the
        optimal cutoff differs from the Youden cutoff;
  \item \textup{(SN3, ST3):} a weak-separation regime in which $\xi_1-\xi_0=1$, producing a flatter estimating equation and a larger asymptotic variance.
\end{itemize}
The skew-$t$ scenarios additionally introduce heavier tails through degrees of freedom $\nu\in\{5,7,8\}$, thereby requiring estimation of an extra tail parameter and increasing the variability of the cutoff estimator relative to the corresponding skew-normal settings.

All simulations were conducted under correctly specified models, so that the fitted SMSN distributions coincide with the data-generating mechanisms, ensuring that discrepancies
from the asymptotic predictions are attributable solely to finite-sample effects. For each scenario, data were generated independently within
each group according to
\[
X_{k1},\ldots,X_{kn_k}
\stackrel{\mathrm{i.i.d.}}{\sim}
f_k(\cdot\,;\theta_k^*),
\quad k\in\{0,1\},
\]
with balanced group sizes $n_0=n_1=n/2$.
The skew-normal distributions were parametrised as
$\theta_k=(\xi_k,\log\omega_k,\alpha_k)$ and the skew-$t$
distributions as
$\theta_k=(\xi_k,\log\omega_k,\alpha_k,\log(\nu_k-2))$,
ensuring $\omega_k>0$ and $\nu_k>2$ throughout the
optimisation.

The true optimal cutoff $c^*$ was computed as the unique
solution of $\varphi(c,\theta^*)=0$ over the effective
support of the data, and the parametric Youden cutoff
$c_Y$ as the solution of $f_1(c;\theta_1^*)=f_0(c;\theta_0^*)$.

For each scenario and sample size, we generated $B=2000$
Monte Carlo replications.
A replication was classified as successful if all numerical
steps of the procedure were completed, namely: successful
maximum likelihood estimation in both groups, successful
solution of the cutoff equation, invertibility of the
regularised observed Hessian matrices, and a finite positive
plug-in variance estimate.
In addition, we required the estimated local slope of the
estimating equation at the cutoff to satisfy
$|\partial_c\varphi(\hat c,\hat\theta)|>10^{-3}$.
This threshold was introduced solely as a numerical
stability safeguard.
Indeed, the plug-in asymptotic variance involves division by
$\bigl[\partial_c\varphi(\hat c,\hat\theta)\bigr]^2$, so
near-zero values of $\partial_c\varphi$ may produce unstable variance estimates due to numerical flatness of the estimating equation.
The threshold $10^{-3}$ is therefore not a substantive model assumption, but a conservative filter excluding only practically degenerate numerical cases. Its empirical impact was negligible. Across the $18$ simulation configurations ($6$ scenarios $\times$ $3$ sample sizes), only $10$ replications out of a total of $36{,}000$ were excluded, corresponding to an overall success rate of $0.9997$.
The largest number of exclusions occurred in scenario ST2 at $n=200$ (four replications), while most configurations had either no exclusions or only a single excluded replication. Accordingly, the column ``Succ.'' in Table~\ref{tab:asymptotic_validation} should be interpreted as a measure of numerical reliability of the complete estimation-and-inference pipeline, rather than as evidence of any substantive failure of the proposed methodology.

\subsubsection{Choice of the admissible interval \texorpdfstring{$\mathcal{C}=[a,b]$}{C=[a,b]}}

The admissible interval $\mathcal C=[a,b]$ is the compact decision region over which the cutoff optimisation is performed. Its role is both theoretical and practical. From a theoretical perspective, compactness guarantees the existence of a minimiser whenever the risk function $c\mapsto R(c;\theta)$ is continuous. From a numerical perspective, restricting the optimisation to a statistically relevant range of marker values improves stability and avoids solutions driven by extreme tail behaviour, where fitted densities may become very small and direct likelihood-ratio evaluation may be unstable.

In the simulation study, the interval $[a,b]$ was constructed from the effective support of the fitted class distributions. More precisely, for each group $k\in\{0,1\}$, we computed the lower and
upper quantiles $q_{k,\alpha/2}$ and $q_{k,1-\alpha/2}$ of the corresponding fitted SMSN model, with $\alpha=0.01$, and defined
\[
a=\min(q_{0,\alpha/2},\,q_{1,\alpha/2}),
\quad
b=\max(q_{0,1-\alpha/2},\,q_{1,1-\alpha/2}).
\]
Thus, the optimisation region contains the central $99\%$ effective
support of both fitted distributions and adapts automatically to
differences in location, scale, skewness, and tail behaviour across
scenarios.

After constructing this initial interval, the boundary conditions were checked numerically through the estimating function rather than through the likelihood ratio itself. Specifically, we verified whether
\[
\varphi(a;\hat\theta)<0
\qquad\text{and}\qquad
\varphi(b;\hat\theta)>0,
\]
which ensures that the root of the estimating equation
$\varphi(c,\hat\theta)=0$ is bracketed inside $[a,b]$. This criterion is equivalent to the corresponding likelihood-ratio condition under the monotone likelihood-ratio setting, but is substantially more
stable numerically because it avoids direct division of two very small tail densities. Whenever one of the above inequalities failed, the interval was widened iteratively until the sign change was achieved.

This construction proved numerically stable in the simulation study.
As reported in Table~\ref{tab:asymptotic_validation}, the bracketing rate was equal or very close to one in all regular scenarios, and only the most challenging heavy-tailed weak-separation setting required
substantial interval expansion in smaller samples. In the serological
application (Section~\ref{sec:application}), by contrast, the admissible interval was determined from extreme quantiles of the log$_{10}$-transformed observations, as described in Section~\ref{sec:decision-theoretic-cutoff-estimation}.

\subsection{Monte Carlo simulation and performance measures}

For each scenario, total sample sizes
$n\in\{200,400,1000\}$ were considered, and
$B=2000$ Monte Carlo replications were performed.
In each replication, the SMSN model was fitted separately
in each group by maximum likelihood.
To ensure numerical accuracy of the plug-in variance
estimator, the observed Fisher information matrices were
computed via Richardson extrapolation as implemented in the
\texttt{numDeriv} package \citep{numDeriv}, rather than
the default finite-difference approximation returned by
the optimiser; this choice is critical for the reliability
of $\hat{V}$ in heavy-tailed models whose log-likelihood
involves numerical integration.
The plug-in cutoff $\hat{c}$ was obtained as the unique
solution of $\varphi(c,\hat\theta)=0$, and the plug-in asymptotic variance was computed as in equation \eqref{eq:V-hat}.
All gradients were evaluated numerically using Richardson extrapolation.

A replication was classified as successful if:
(i) the estimating equation exhibited a unique sign change in a neighbourhood of the solution;
(ii) the observed information matrices were numerically invertible; and
(iii) $|\partial_c\varphi(\hat{c},\hat\theta)|$ exceeded a tolerance of $10^{-3}$.
Replications failing any of these criteria were excluded from the asymptotic summaries.
Success rates were uniformly high across all scenarios and sample sizes, confirming the procedure's numerical stability under well-conditioned configurations.

The following quantities are reported in
Table~\ref{tab:asymptotic_validation}:
the empirical mean $\bar{\hat{c}}$, bias, root mean squared error (RMSE) and standard deviation (SD) of the estimator; the empirical variance of $\sqrt{n}(\hat{c}-c^*)$
(Var$_\mathrm{emp}$) and its average plug-in estimate
(Var$_\mathrm{th}$), together with their ratio
\[
\text{Ratio}
=
\frac{
  \mathrm{Var}\bigl\{\sqrt{n}(\hat{c}-c^*)\bigr\}
}{
  \mathbb{E}[\hat{V}]
};
\]
the mean local identifiability diagnostic
$|\partial_c\varphi|$; and the empirical coverage and
average length of the $95\%$ Wald confidence interval $\hat{c}\pm z_{0.975}\sqrt{\hat{V}/n}$.
With $B=2000$ replications, the empirical coverage is itself subject
to Monte Carlo error. Under a nominal coverage probability of $0.95$,
its standard error is $\sqrt{0.95(1-0.95)/2000}\approx 0.005$.
Thus, observed coverages within approximately $\pm 0.01$ of $0.95$,
that is, in the interval $[0.940,0.960]$, are compatible with a
well-calibrated $95\%$ procedure.

\begin{table*}[h!]
\centering
\caption{Monte Carlo validation of the asymptotic distribution of the optimal cutoff estimator under skew-normal (SN) and skew-$t$ (ST) distributions, based on $B=2000$ simulated samples per configuration. For each scenario and sample size, we report the empirical mean ($\bar{\hat{c}}$), bias, root mean squared error (RMSE), and standard deviation (SD) of the estimator. The columns Var$_{\mathrm{emp}}$ and Var$_{\mathrm{th}}$ denote, respectively, the empirical variance of $\sqrt{n}(\hat c-c^*)$ and the average plug-in asymptotic variance estimate, while their ratio assesses the accuracy of the asymptotic approximation. The quantity $|\partial_c\varphi|$ is reported through its empirical mean and summarises the local identifiability of the cutoff. Wald confidence intervals are summarised by their empirical coverage probability (Cov.) and average length (Len.). The column Succ. reports the proportion of successful Monte Carlo replications, and Bracket reports the proportion of replications in which the numerically constructed admissible interval satisfied the bracketing condition $\varphi(a;\hat\theta)<0<\varphi(b;\hat\theta)$.}
\label{tab:asymptotic_validation}

\scriptsize
\setlength{\tabcolsep}{3pt}
\renewcommand{\arraystretch}{0.9}

\begin{tabular*}{\textwidth}{@{\extracolsep\fill}lcccccccccccccc@{\extracolsep\fill}}
\toprule
\textbf{Scenario} & $c^*$ & $n$ & Succ. & $\bar{\hat{c}}$ & Bias & RMSE & SD & Var$_{\text{emp}}$ & Var$_{\text{th}}$ & Ratio & $|\partial_c \varphi|$ & Cov. & Len. & Bracket \\
\midrule

\multirow{3}{*}{SN1} & \multirow{3}{*}{1.6101}
& 200  & 1.0000 & 1.6103 &  0.0002 & 0.0636 & 0.0636 & 0.8083 & 0.8043 & 1.0049 & 0.4009 & 0.9480 & 0.2477 & 1.0000 \\
& & 400  & 0.9995 & 1.6098 & -0.0003 & 0.0453 & 0.0453 & 0.8216 & 0.8043 & 1.0215 & 0.3884 & 0.9435 & 0.1755 & 1.0000 \\
& & 1000 & 1.0000 & 1.6106 &  0.0005 & 0.0285 & 0.0285 & 0.8114 & 0.8150 & 0.9956 & 0.3829 & 0.9475 & 0.1118 & 1.0000 \\

\midrule

\multirow{3}{*}{SN2} & \multirow{3}{*}{1.7607}
& 200  & 1.0000 & 1.7586 & -0.0021 & 0.0673 & 0.0673 & 0.9060 & 0.8557 & 1.0588 & 0.5412 & 0.9410 & 0.2552 & 1.0000 \\
& & 400  & 1.0000 & 1.7607 &  0.0000 & 0.0467 & 0.0467 & 0.8718 & 0.8462 & 1.0303 & 0.5110 & 0.9430 & 0.1800 & 1.0000 \\
& & 1000 & 1.0000 & 1.7613 &  0.0007 & 0.0294 & 0.0294 & 0.8656 & 0.8390 & 1.0317 & 0.5018 & 0.9455 & 0.1135 & 1.0000 \\

\midrule

\multirow{3}{*}{SN3} & \multirow{3}{*}{1.0122}
& 200  & 0.9995 & 1.0070 & -0.0052 & 0.1001 & 0.1000 & 1.9983 & 1.9277 & 1.0366 & 0.3531 & 0.9415 & 0.3784 & 0.9965 \\
& & 400  & 0.9985 & 1.0106 & -0.0016 & 0.0672 & 0.0672 & 1.8070 & 1.8069 & 1.0001 & 0.3440 & 0.9449 & 0.2613 & 1.0000 \\
& & 1000 & 1.0000 & 1.0095 & -0.0027 & 0.0420 & 0.0419 & 1.7551 & 1.7475 & 1.0043 & 0.3393 & 0.9475 & 0.1635 & 1.0000 \\

\midrule

\multirow{3}{*}{ST1} & \multirow{3}{*}{1.6145}
& 200  & 1.0000 & 1.6185 &  0.0040 & 0.0758 & 0.0757 & 1.1459 & 1.1638 & 0.9846 & 0.3605 & 0.9430 & 0.2973 & 0.9930 \\
& & 400  & 0.9995 & 1.6134 & -0.0011 & 0.0551 & 0.0551 & 1.2162 & 1.1604 & 1.0481 & 0.3510 & 0.9395 & 0.2107 & 1.0000 \\
& & 1000 & 1.0000 & 1.6155 &  0.0010 & 0.0346 & 0.0346 & 1.1971 & 1.1624 & 1.0299 & 0.3456 & 0.9515 & 0.1336 & 1.0000 \\

\midrule

\multirow{3}{*}{ST2} & \multirow{3}{*}{1.7842}
& 200  & 0.9980 & 1.7840 & -0.0002 & 0.0830 & 0.0830 & 1.3783 & 1.3011 & 1.0594 & 0.4739 & 0.9389 & 0.3136 & 0.9950 \\
& & 400  & 1.0000 & 1.7839 & -0.0003 & 0.0558 & 0.0558 & 1.2460 & 1.2731 & 0.9787 & 0.4560 & 0.9490 & 0.2207 & 1.0000 \\
& & 1000 & 1.0000 & 1.7846 &  0.0004 & 0.0359 & 0.0359 & 1.2902 & 1.2629 & 1.0216 & 0.4495 & 0.9435 & 0.1392 & 1.0000 \\

\midrule

\multirow{3}{*}{ST3} & \multirow{3}{*}{0.9723}
& 200  & 1.0000 & 0.9695 & -0.0028 & 0.1024 & 0.1024 & 2.0976 & 2.2643 & 0.9264 & 0.3317 & 0.9505 & 0.4062 & 0.8240 \\
& & 400  & 1.0000 & 0.9692 & -0.0031 & 0.0717 & 0.0717 & 2.0542 & 2.0535 & 1.0003 & 0.3318 & 0.9445 & 0.2783 & 0.9160 \\
& & 1000 & 1.0000 & 0.9715 & -0.0008 & 0.0447 & 0.0447 & 1.9962 & 1.9576 & 1.0197 & 0.3337 & 0.9485 & 0.1729 & 0.9825 \\

\bottomrule
\end{tabular*}

\begin{tablenotes}
\item {\it Note:} The column Succ.\ reports the proportion of Monte Carlo replications satisfying all numerical validity checks,
including successful model fitting, successful root-finding,
invertible regularised Hessians, and the stability condition
$|\partial_c\varphi(\hat c,\hat\theta)|>10^{-3}$. Var$_{\mathrm{emp}}$/Var$_{\mathrm{th}}$ assesses the accuracy of the asymptotic variance approximation; values close to one indicate strong agreement between empirical and theoretical variability. 
\end{tablenotes}
\end{table*}

\subsection{Results}

\subsubsection{Consistency}

Table~\ref{tab:asymptotic_validation} shows that the plug-in estimator $\hat c$ is highly accurate across all scenarios. In every
configuration, the empirical mean $\bar{\hat c}$ is very close to the
true cutoff $c^*$, and the bias remains small in absolute value. In the skew-normal settings, the largest bias is observed in the weak
separation scenario SN3 at $n=200$ (Bias $=-0.0052$), decreasing to
$-0.0016$ at $n=400$ and $-0.0027$ at $n=1000$. In the skew-$t$
settings, the largest bias is recorded in ST1 at $n=200$
(Bias $=0.0040$) and in ST3 at $n=400$ (Bias $=-0.0031$), with no
evidence of systematic distortion. Overall, the estimator is nearly
unbiased in all regular settings and remains well-behaved even in the
most difficult heavy-tailed cases.

The root mean squared error decreases monotonically with $n$ in all
scenarios, confirming the consistency predicted by Theorem~\ref{thm:asymptotics}. For example, in scenario SN1, the RMSE
decreases from $0.0636$ at $n=200$ to $0.0453$ at $n=400$ and $0.0285$ at $n=1000$, while in ST3 it falls from $0.1024$ to $0.0717$ and $0.0447$, respectively. The standard deviation exhibits the same pattern. These results show that the plug-in estimator converges reliably to the true cutoff under both skew-normal and skew-$t$ sampling regimes.

\subsubsection{Asymptotic normality and variance estimation}

The asymptotic normal approximation is assessed through the ratio
Var$_{\mathrm{emp}}$/Var$_{\mathrm{th}}$, where Var$_{\mathrm{emp}}$
denotes the empirical variance of $\sqrt n(\hat c-c^*)$ and
Var$_{\mathrm{th}}$ is the average plug-in asymptotic variance. Values
close to one indicates that the first-order asymptotic theory captures well the sampling variability of the estimator.

\paragraph{Skew-normal scenarios.}
For SN1, the ratio remains extremely close to one across all sample
sizes, ranging from $0.996$ to $1.022$, which indicates excellent
agreement between empirical and theoretical variability even at
$n=200$. Scenario SN2 also shows very satisfactory performance, with
ratios in $[1.032,1.059]$. The only noticeable upward deviation occurs
at $n=200$, where the empirical variance is about $6\%$ larger than the
average plug-in variance, but this difference is small and does not
persist as $n$ increases. In the weak-separation scenario SN3, the
ratios are $1.037$, $1.000$, and $1.004$ for $n=200$, $400$, and
$1000$, respectively. Thus, even in the least regular skew-normal case,
the asymptotic variance approximation is remarkably accurate.

\paragraph{Skew-$t$ scenarios.}
The skew-$t$ scenarios exhibit larger asymptotic variances than the
corresponding skew-normal settings, reflecting the added uncertainty
induced by heavier tails and the estimation of the tail parameter. Even
so, the asymptotic approximation remains highly reliable. In ST1, the
ratio lies in $[0.985,1.048]$; in ST2, it lies in $[0.979,1.059]$; and
in ST3, it lies in $[0.926,1.020]$. The greatest discrepancy is
observed in ST3 at $n=200$, where the ratio equals $0.926$. This is the
most challenging configuration, combining weak separation and heavy
tails, and it is precisely the setting in which the local slope
diagnostic $|\partial_c\varphi|$ is smallest. Still, even there the
difference between empirical and theoretical variance remains moderate,
and the approximation improves quickly as $n$ increases.

These findings are fully consistent with
Remark~\ref{rem:variance-interpretation}. Scenarios with smaller values
of $|\partial_c\varphi|$ --- especially SN3 and ST3 --- display larger
RMSE, larger asymptotic variance, and slightly greater finite-sample
deviations from the first-order approximation. Conversely, when the estimating equation is steeper near the optimum, the cutoff is more
sharply identified and the asymptotic variance formula is especially
accurate.

\subsubsection{Coverage of Wald confidence intervals}

The empirical coverage of the nominal $95\%$ Wald confidence intervals is
globally satisfactory. Across all scenarios and sample sizes, the
coverage probabilities remain close to the nominal level, with most
values lying between $0.939$ and $0.952$. In the skew-normal scenarios,
coverage is very stable: SN1 ranges from $0.944$ to $0.948$, SN2 from
$0.941$ to $0.946$, and SN3 from $0.942$ to $0.948$. In the skew-$t$
scenarios, the corresponding ranges are $0.940$ to $0.952$ for ST1,
$0.939$ to $0.949$ for ST2, and $0.945$ to $0.951$ for ST3. Hence, the
Wald intervals perform well not only in regular skew-normal settings
but also in the more demanding heavy-tailed cases.

The average interval length decreases systematically with $n$, in
agreement with the $n^{-1/2}$ rate implied by the asymptotic theory. In
scenario SN1, for instance, the mean length decreases from $0.2477$ at
$n=200$ to $0.1755$ at $n=400$ and $0.1118$ at $n=1000$. Similarly, in
ST2 the mean length decreases from $0.3136$ to $0.2207$ and $0.1392$,
and in ST3 from $0.4062$ to $0.2783$ and $0.1729$. As expected, the
widest intervals occur in the weakest-identifiability settings, again
particularly in ST3, where heavier tails and smaller local slope values
inflate the uncertainty of the cutoff estimator.

\subsubsection{Numerical stability and bracketing of the estimating equation}

The numerical implementation was highly stable throughout the
simulation study. The proportion of successful Monte Carlo replications
was equal or extremely close to one in all scenarios, with the lowest
success rate still above $0.998$. This indicates that separate SMSN
maximum likelihood fitting, root-finding, and plug-in variance
evaluation are computationally reliable under the proposed algorithm.

The numerical bracketing condition for the estimating equation,
summarised in the column \textit{Bracket}, was satisfied in essentially
all replications for the skew-normal scenarios and for the regular
skew-$t$ settings ST1 and ST2, especially once $n\geq 400$. The only
substantial difficulty arose in ST3, where the bracketing rate was
$0.824$ at $n=200$, $0.916$ at $n=400$, and $0.983$ at $n=1000$. This pattern is again consistent with weaker local identifiability in the
heavy-tailed weak-separation regime: the flatter estimating equation
makes numerical localisation of the root more delicate in smaller
samples, although this difficulty largely disappears as the sample size
increases.

Overall, the simulation study confirms that the proposed plug-in
estimator, its asymptotic variance formula, and the associated Wald
inference perform reliably across the full range of distributional
settings considered here.

\subsubsection{Normal Q--Q plots}

Figure~\ref{fig:qqplots-scenarios} presents normal Q--Q plots of the
standardised statistic
\[
W_n
=
\frac{\sqrt{n}(\hat c-c^*)}{\sqrt{\bar{\hat V}}},
\]
where $\bar{\hat V}$ denotes the average plug-in asymptotic variance
across Monte Carlo replications. Under the asymptotic theory developed
in Theorem~\ref{thm:asymptotics}, $W_n$ should be approximately
standard normal.

The graphical evidence is fully consistent with the numerical results
reported in Table~\ref{tab:asymptotic_validation}. In the regular
skew-normal scenarios SN1 and SN2, the empirical quantiles lie very
close to the reference normal line at all sample sizes, including
$n=200$, which confirms the excellent accuracy of the first-order
asymptotic approximation already in moderate samples. In the
weak-separation scenario SN3, mild deviations from normality are still
visible at $n=200$, particularly in the tails, but these discrepancies
become much less pronounced at $n=400$ and are small at $n=1000$.
This behaviour mirrors the slightly greater finite-sample difficulty of
that scenario, as reflected in its lower local slope values and larger
sampling variability.

A similar pattern is observed in the skew-$t$ settings. In ST1 and ST2, the Q--Q plots show good agreement with the normal reference line,
with only modest tail dispersion in the smallest sample size and clear
improvement as $n$ increases. The most challenging case is ST3, which
combines heavy tails with weak separation. There, departures from the
reference line are more visible at $n=200$, but the alignment improves
substantially at $n=400$ and becomes much closer to linearity at
$n=1000$. Overall, the Q--Q plots provide visual confirmation that the
scaled cutoff estimator converges to normality across all scenarios,
and that the approximation improves systematically with sample size,
especially in the more weakly identifiable settings.

\begin{figure*}[h!]
\centering
\includegraphics[width=\textwidth]{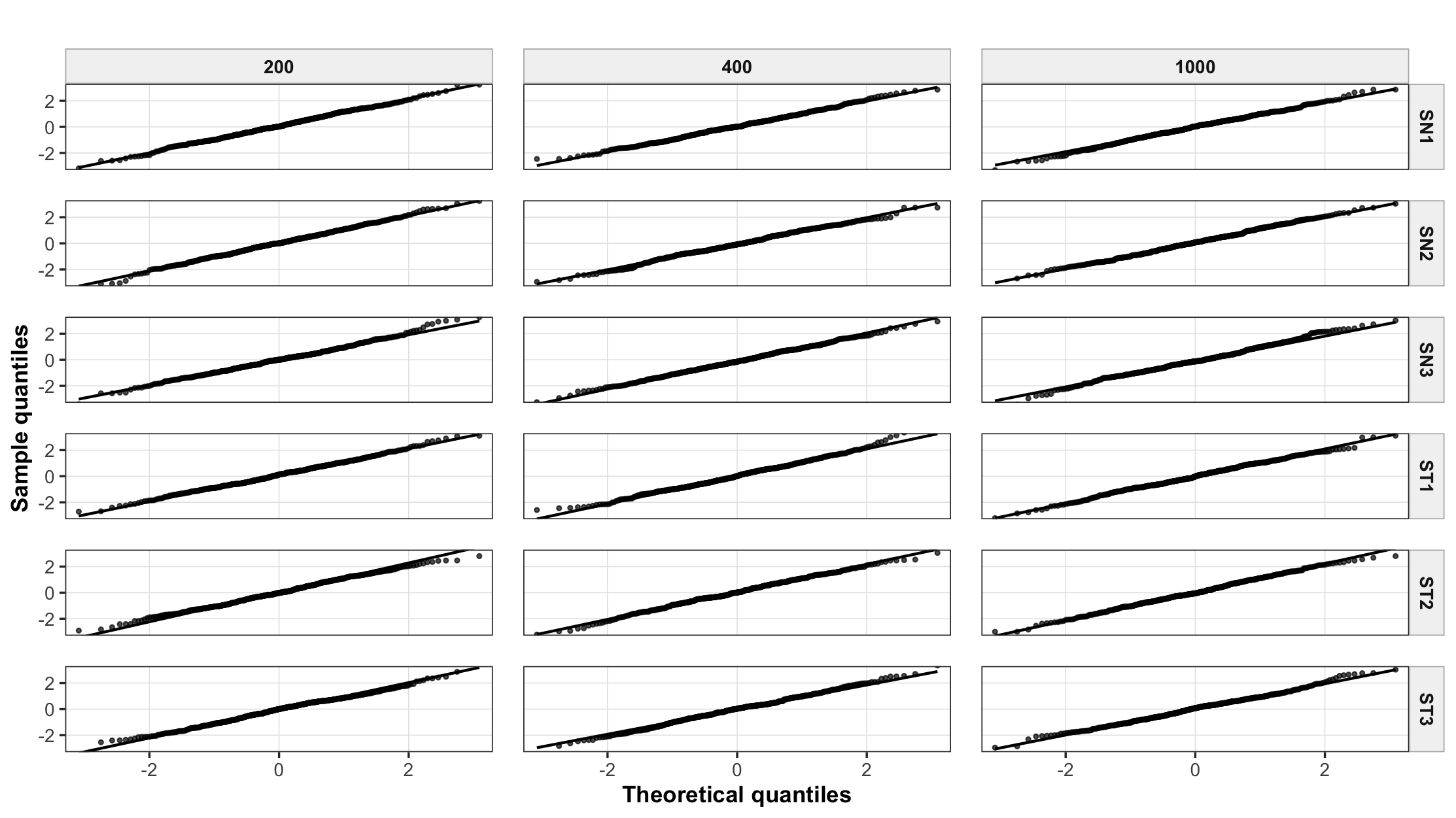}
\caption{Normal Q--Q plots of the standardised scaled cutoff estimator
$W_n = \sqrt{n}(\hat{c} - c^*) \big/ \sqrt{\bar{\hat{V}}}$,
where $\bar{\hat{V}}$ denotes the mean of the plug-in asymptotic
variance estimates $\hat{V}$ across the $B = 2000$ Monte Carlo
replications. Under the asymptotic theory of Theorem~\ref{thm:asymptotics} , $W_n$ converges in
distribution to a standard normal variable. The plots show excellent agreement between empirical and theoretical quantiles in the skew-normal scenarios (SN1, SN2, SN3) at all sample sizes, with only minor departures in the lower tail of SN3 at $n = 200$ that disappear as $n$ increases.
Among the skew-$t$ scenarios, ST1 and ST3 exhibit satisfactory alignment
at all sample sizes, while ST2 shows mild heavy-tail deviations at
$n = 200$ that progressively diminish as $n$ increases, consistent
with the larger asymptotic variances and the additional estimation
uncertainty associated with the tail parameter $\nu$.}
\label{fig:qqplots-scenarios}
\end{figure*}

\subsubsection{Role of the local identifiability diagnostic}

A structural feature of the asymptotic variance formula $V$ is that estimation precision is governed jointly by the uncertainty in $\hat\theta$, encoded in $\Sigma$, and by the local slope of the estimating equation at the decision boundary, encoded in $|\partial_c\varphi(c^*,\theta^*)|$.
As established in Remark~\ref{rem:variance-interpretation}, a smaller
slope corresponds to a flatter boundary and inflates the
asymptotic variance, while a steeper slope yields a more stably identified cutoff.

This inverse relationship is clearly reflected in the
simulation results.
The weak-separation scenarios SN3 and ST3 exhibit the
smallest values of $|\partial_c\varphi|$, approximately
$0.33$--$0.36$, and correspondingly the largest asymptotic
variances, in the range $1.73$--$2.26$.
In contrast, the asymmetric scenarios SN2 and ST2, which
have the largest values of $|\partial_c\varphi|$,
approximately $0.45$--$0.54$, produce the smallest
variances, in the range $0.84$--$1.30$.
The monotone inverse relationship between
$|\partial_c\varphi|$ and estimation variability is
consistent across all six scenarios and provides direct
empirical corroboration of the theoretical interpretation
advanced in Remark~\ref{rem:variance-interpretation}.
These findings support the use of
$|\partial_c\varphi(\hat{c},\hat\theta)|$ as a routine
diagnostic for the reliability of the estimated decision
boundary: values close to zero signal weak local
identifiability and should prompt the use of bootstrap
confidence intervals as a more robust alternative to the
Wald interval.

\section{Application to serological data}
\label{sec:application}

\subsection{Data and model selection}
\label{sec:data-model-selection}

We apply the proposed methodology to serological data from
a SARS-CoV-2 immunological study \cite{rosado2020multiplex}, in which antibody concentrations against four antigens --- Stri, RBD, S1, and
S2 --- were measured by IgG assay in a sample of
seronegative ($D=0$) and seropositive ($D=1$) individuals. A detailed description of lab procedures can be found in the original study \cite{rosado2020multiplex}.
All biomarker values were analysed on the $\log_{10}$ scale,
and model selection was performed separately within each
group using the Bayesian Information Criterion (BIC) to
choose between a skew-normal and a skew-$t$ distribution. 

The results, reported in Table~\ref{tab:model-selection-igg},
reveal a consistent pattern: the seronegative group is
uniformly better described by a skew-$t$ distribution across
all four biomarkers, with $\Delta\mathrm{BIC}=
\mathrm{BIC}_\mathrm{SN}-\mathrm{BIC}_\mathrm{ST}$
ranging from $11.55$ (S2\_IgG) to $64.41$ (RBD\_IgG),
providing strong to very strong evidence in favour of the
heavier-tailed model for this group.
The seropositive group, by contrast, is consistently
better described by a skew-normal distribution, with
$\Delta\mathrm{BIC}$ negative in all four cases
(range: $-1.13$ to $-5.37$).
This asymmetry is biologically interpretable: seropositive
individuals exhibit a broad distribution of antibody
concentrations reflecting heterogeneous immune responses,
whereas the seronegative distribution is more peaked but
with heavier tails arising from non-specific low-level
reactivity and measurement variation.

Accordingly, for each biomarker, all subsequent analyses
use a skew-$t$ model for the seronegative group and a
skew-normal model for the seropositive group.

\begin{table}[ht]
\centering
\caption{Model selection between skew-normal (SN) and skew-$t$ (ST)
distributions for IgG biomarkers, performed separately within each group
using BIC. Positive $\Delta\mathrm{BIC}$ favours the ST model; negative
values favour the SN model.}
\label{tab:model-selection-igg}
\small
\begin{tabular}{@{}llcccc@{}}
\toprule
\textbf{Biomarker} & \textbf{Group} & \textbf{Selected} &
  $\mathrm{BIC}_\mathrm{SN}$ & $\mathrm{BIC}_\mathrm{ST}$ &
  $\Delta\mathrm{BIC}$ \\
\midrule
\multirow{2}{*}{Stri\_IgG}
& Seronegative & ST &  124.86 &  90.07 & 34.78 \\
& Seropositive & SN &  273.15 & 274.28 & -1.13 \\
\midrule
\multirow{2}{*}{RBD\_IgG}
& Seronegative & ST &   17.99 & -46.42 & 64.41 \\
& Seropositive & SN &  342.39 & 347.75 & -5.37 \\
\midrule
\multirow{2}{*}{S1\_IgG}
& Seronegative & ST & -160.57 & -198.92 & 38.35 \\
& Seropositive & SN &  306.73 &  312.10 & -5.37 \\
\midrule
\multirow{2}{*}{S2\_IgG}
& Seronegative & ST &   -3.64 &  -15.19 & 11.55 \\
& Seropositive & SN &  326.03 &  331.40 & -5.37 \\
\bottomrule
\end{tabular}
\end{table}

\subsection{Decision-theoretic cutoff estimation}
\label{sec:decision-theoretic-cutoff-estimation}

\subsubsection{Choice and numerical verification of the admissible interval \texorpdfstring{$[a,b]$}{[a,b]} in the serological application}
For the serological data analysis, the admissible interval $[a,b]$
used for cutoff optimisation was constructed directly from the
observed biomarker values on the $\log_{10}$ scale. More precisely,
for each biomarker, we pooled the seronegative and seropositive
observations and defined an initial interval from extreme empirical
quantiles of the pooled sample, namely the $0.5$th and $99.5$th
percentiles. This choice yields a data-driven optimisation region
covering almost all observed values while limiting the influence of
extreme observations and numerical instability in the far tails.

After fitting the group-specific models, the
resulting interval was checked numerically through the estimating
function rather than through the likelihood ratio itself. Specifically,
for each biomarker we verified that $\varphi(a;\hat\theta)<0$ and $\varphi(b;\hat\theta)>0.$
This sign change guarantees that the root of the estimating equation
$\varphi(c,\hat\theta)=0$ is bracketed inside $[a,b]$. For all four
IgG biomarkers analysed here, the initial interval satisfied this
condition without requiring iterative expansion, indicating that the
empirical quantile construction was sufficient to localise the
decision-relevant region of the biomarker distribution.

\begin{table}[ht]
\centering
\caption{Estimated optimal cutoff $\hat{c}^*$, parametric and empirical
Youden cutoffs, associated risks, 95\% Wald confidence intervals,
local slope diagnostic, and admissible interval diagnostics for the
four SARS-CoV-2 IgG biomarkers under
$(\lambda_0,\lambda_1;\pi_0,\pi_1) = (1,3;0.9,0.1)$. Cutoffs and
CIs are reported on the original scale.}
\label{tab:real-data-IgG-cutoffs-interval}
\tiny
\setlength{\tabcolsep}{2pt}
\renewcommand{\arraystretch}{1.0}
\begin{tabular}{@{}lccccccccccccc@{}}
\toprule
\textbf{Biomarker} & $f_0$ & $f_1$ & $\hat{c}^*$ & $\hat{c}_Y$ &
  $\hat{c}_Y^\mathrm{emp}$ & $\widehat{R}(\hat{c}^*)$ &
  $\widehat{R}(\hat{c}_Y)$ & $\widehat{R}(\hat{c}_Y^\mathrm{emp})$ &
  CI$_{95\%}$ & $|\partial_c\varphi|$ & $a$ & $b$ & Bracket \\
\midrule
Stri\_IgG & ST & SN & 342.49 & 211.24 & 468.0 & 0.0567 & 0.0650 & 0.0596 & [339.38, 345.61] & 0.3347 & 1.2641 & 3.9560 & Yes \\
RBD\_IgG  & ST & SN & 305.73 & 191.03 & 140.0 & 0.0521 & 0.0598 & 0.0771 & [302.86, 308.64] & 0.3066 & 1.3686 & 4.1446 & Yes \\
S1\_IgG   & ST & SN & 141.30 & 106.16 & 116.5 & 0.0884 & 0.1039 & 0.0949 & [140.24, 142.36] & 1.5701 & 1.4431 & 3.6415 & Yes \\
S2\_IgG   & ST & SN & 264.82 & 179.71 & 168.0 & 0.0806 & 0.0935 & 0.0992 & [263.01, 266.65] & 0.7436 & 1.4344 & 3.8999 & Yes \\
\bottomrule
\end{tabular}
\medskip\\
{\footnotesize\textit{Note:} $f_0$ and $f_1$ are the fitted distributions in the
seronegative and seropositive groups. The quantities $a$ and $b$ are
the lower and upper bounds of the admissible interval on the
$\log_{10}$ scale. Bracket indicates whether the sign condition
$\varphi(a;\hat\theta) < 0 < \varphi(b;\hat\theta)$ was satisfied
without further expansion.}
\end{table}

The optimal cutoff $\hat c^*$ was estimated as the plug-in
solution of $\varphi(c,\hat\theta)=0$, where $\hat\theta$
collects the separate maximum likelihood estimators obtained
under a skew-$t$ model for the seronegative group and a
skew-normal model for the seropositive group.
Wald $95\%$ confidence intervals were constructed from the
asymptotic variance formula in
Proposition~\ref{prop:separate-mle}, with numerical
derivatives and observed information computed by Richardson
extrapolation.
For comparison, the parametric Youden cutoff $\hat c_Y$ was
obtained as the solution of
$f_1(c;\hat\theta_1)=f_0(c;\hat\theta_0)$, whereas the
empirical Youden cutoff $\hat c_Y^{\mathrm{emp}}$ was
computed directly from the sample as the maximiser of the
empirical Youden index as in \eqref{eq:youden-empirical}.

Results for the four IgG biomarkers under the asymmetric
decision setting
$(\lambda_0,\lambda_1;\pi_0,\pi_1)=(1,3;0.9,0.1)$ are
reported in Table~\ref{tab:real-data-IgG-cutoffs-interval}.
Since
\[
\frac{\lambda_0\pi_0}{\lambda_1\pi_1}=3,
\]
the optimal cutoff is expected to exceed the Youden cutoff,
and this is observed for all biomarkers.
On the original scale, $\hat c^*$ ranges from $141.30$
for S1\_IgG to $342.49$ for Stri\_IgG, while the
corresponding estimated risks range from $0.0521$
for RBD\_IgG to $0.0884$ for S1\_IgG.

For all four biomarkers, $\hat c^*$ yields lower estimated
risk than both $\hat c_Y$ and
$\hat c_Y^{\mathrm{emp}}$.
Relative to the parametric Youden rule, the reduction in
estimated risk ranges from $0.0077$ for RBD\_IgG to
$0.0155$ for S1\_IgG; relative to the empirical Youden
rule, it ranges from $0.0029$ for Stri\_IgG to $0.0250$,
again with the largest gain observed for RBD\_IgG.
The local slope diagnostic
$|\partial_c\varphi(\hat c^*,\hat\theta)|$ varies
substantially across biomarkers, from $0.3066$
for RBD\_IgG to $1.5701$ for S1\_IgG, and this
variation is reflected in the confidence intervals.
The narrowest interval is obtained for S1\_IgG,
$[140.24,\,142.36]$, whereas wider intervals occur for
biomarkers with smaller local slope, such as RBD\_IgG
($[302.86,\,308.64]$) and Stri\_IgG
($[339.38,\,345.61]$).

The admissible interval for numerical optimisation was
constructed from extreme empirical quantiles of the pooled
$\log_{10}$-transformed observations.
For all four biomarkers, the resulting interval satisfied
$\varphi(a;\hat\theta)<0<\varphi(b;\hat\theta)$ without
further expansion, indicating that the decision-relevant
root was successfully bracketed in every case.

\subsection{Sensitivity analysis for the RBD\_IgG biomarker}

To illustrate the effect of cost and prevalence asymmetry on
the optimal threshold, we examine the RBD\_IgG biomarker in
detail under four representative decision-theoretic
scenarios.
RBD\_IgG was selected because it is particularly informative
from both a practical and a methodological perspective:
among the four biomarkers in
Table~\ref{tab:real-data-IgG-cutoffs-interval}, it attains
the lowest estimated optimal risk
($\widehat R(\hat c^*)=0.0521$), the largest reduction in
estimated risk relative to the empirical Youden cutoff
($0.0250$), and the smallest local slope diagnostic
($|\partial_c\varphi(\hat c^*,\hat\theta)|=0.3066$).
The fitted distributions and decision thresholds under the
four scenarios are shown in
Figure~\ref{fig:rbd-densities}, and the corresponding
numerical results are reported in
Table~\ref{tab:RBD_decision}.

The results highlight two complementary features of the
proposed framework.
First, the ordering of the optimal cutoffs across scenarios
is determined by the monotonicity of the likelihood ratio
$\Lambda(c;\theta)$.
Since the optimal cutoff is the unique solution of
$\Lambda(c^*;\theta)=\frac{\lambda_0\pi_0}{\lambda_1\pi_1}$,
and $\Lambda(\cdot;\theta)$ is strictly increasing, larger
cost--prevalence ratios imply larger optimal cutoffs.
Hence the ordering $c_A^*<c_C^*<c_B^*<c_D^*$ that appears below
follows directly from the ordering of the target ratios
across the four scenarios.

Second, Proposition~\ref{prop:youden-relationship} remains
useful as a local interpretation of the displacement from
the Youden cutoff under mild asymmetry.
It explains why, near the symmetric setting, the direction
of the shift is governed by the sign of
$\log\!\left(\frac{\lambda_0\pi_0}{\lambda_1\pi_1}\right)$.

However, this first-order expansion is local and should not
be interpreted as a quantitative approximation in strongly
asymmetric settings such as Scenario~D, where the target
ratio is far from the symmetric benchmark.
In that case, the relevant justification is the global
monotonicity result of
Proposition~\ref{prop:optimal-cutoff}.

\begin{figure*}[h!]
\centering
\includegraphics[width=\textwidth]{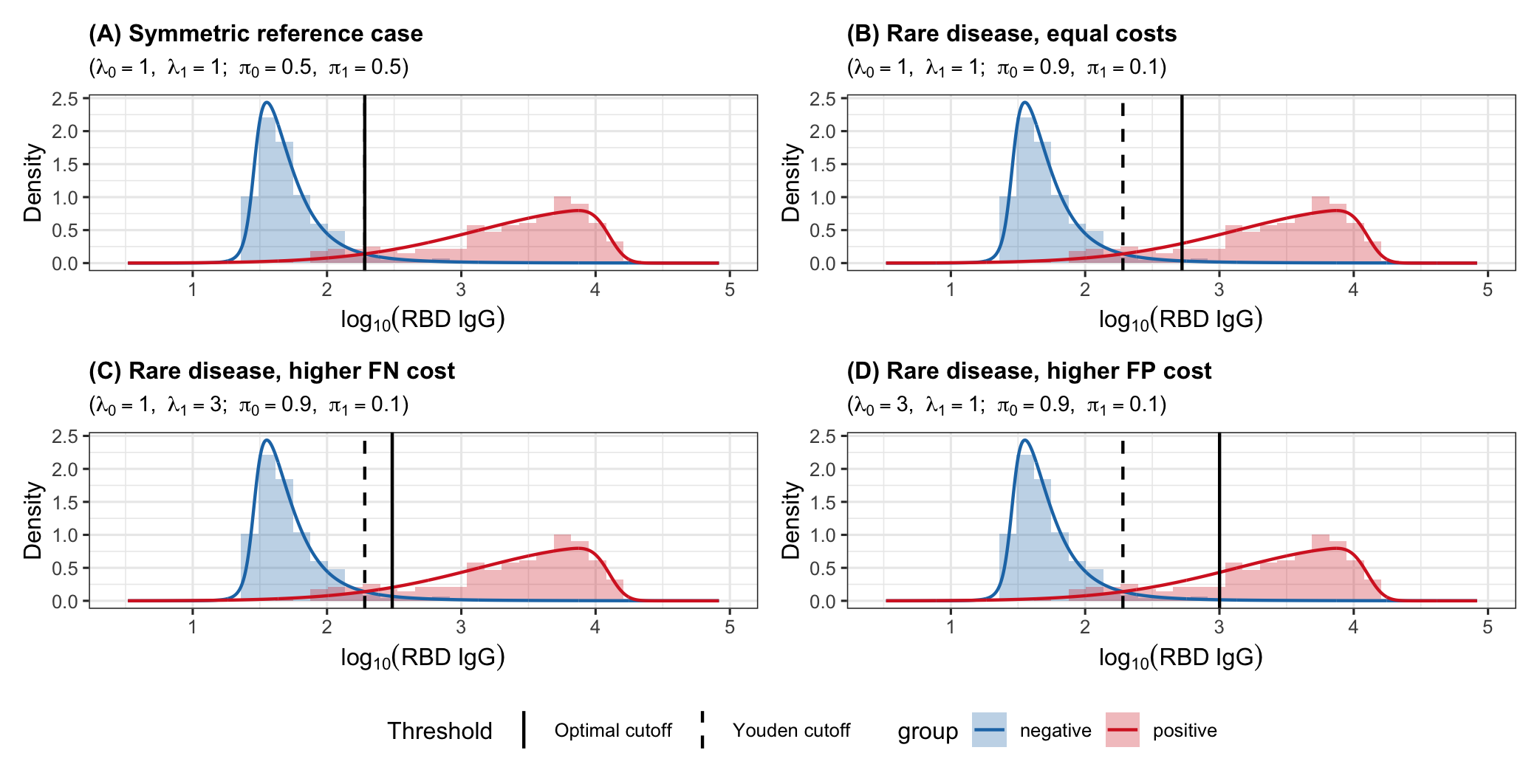}
\caption{
Estimated biomarker distributions and decision thresholds for the RBD\_IgG antibody under four decision-theoretic scenarios. Solid curves represent the fitted skew-$t$ distribution for seronegative individuals and skew-normal distribution for seropositive individuals, while histograms show the empirical distributions on the $\log_{10}$ scale. The vertical solid line denotes the optimal cutoff $c^*$ obtained by minimizing the weighted misclassification risk, and the dashed line denotes the Youden cutoff $c_Y$. Panels correspond to different combinations of misclassification costs $(\lambda_0,\lambda_1)$ and prevalences $(\pi_0,\pi_1)$:
(A) symmetric reference case $(1,1; 0.5,0.5)$;
(B) rare disease with equal costs $(1,1; 0.9,0.1)$;
(C) rare disease with higher false-negative cost $(1,3; 0.9,0.1)$;
(D) rare disease with higher false-positive cost $(3,1; 0.9,0.1)$.
The results illustrate how departures from the symmetric setting induce systematic shifts in the optimal cutoff relative to the Youden solution, reflecting the influence of prevalence and cost asymmetry on decision-making.
}
\label{fig:rbd-densities}
\end{figure*}

Figure~\ref{fig:rbd-roc-tangents} provides the corresponding ROC-based geometric interpretation for the four decision-theoretic scenarios. In each case, the optimal cutoff is identified by the point on the ROC curve whose tangent slope equals $\lambda_0\pi_0/(\lambda_1\pi_1)$, whereas the Youden cutoff corresponds to unit slope.

\begin{figure*}[h!]
\centering
\includegraphics[width=\textwidth]{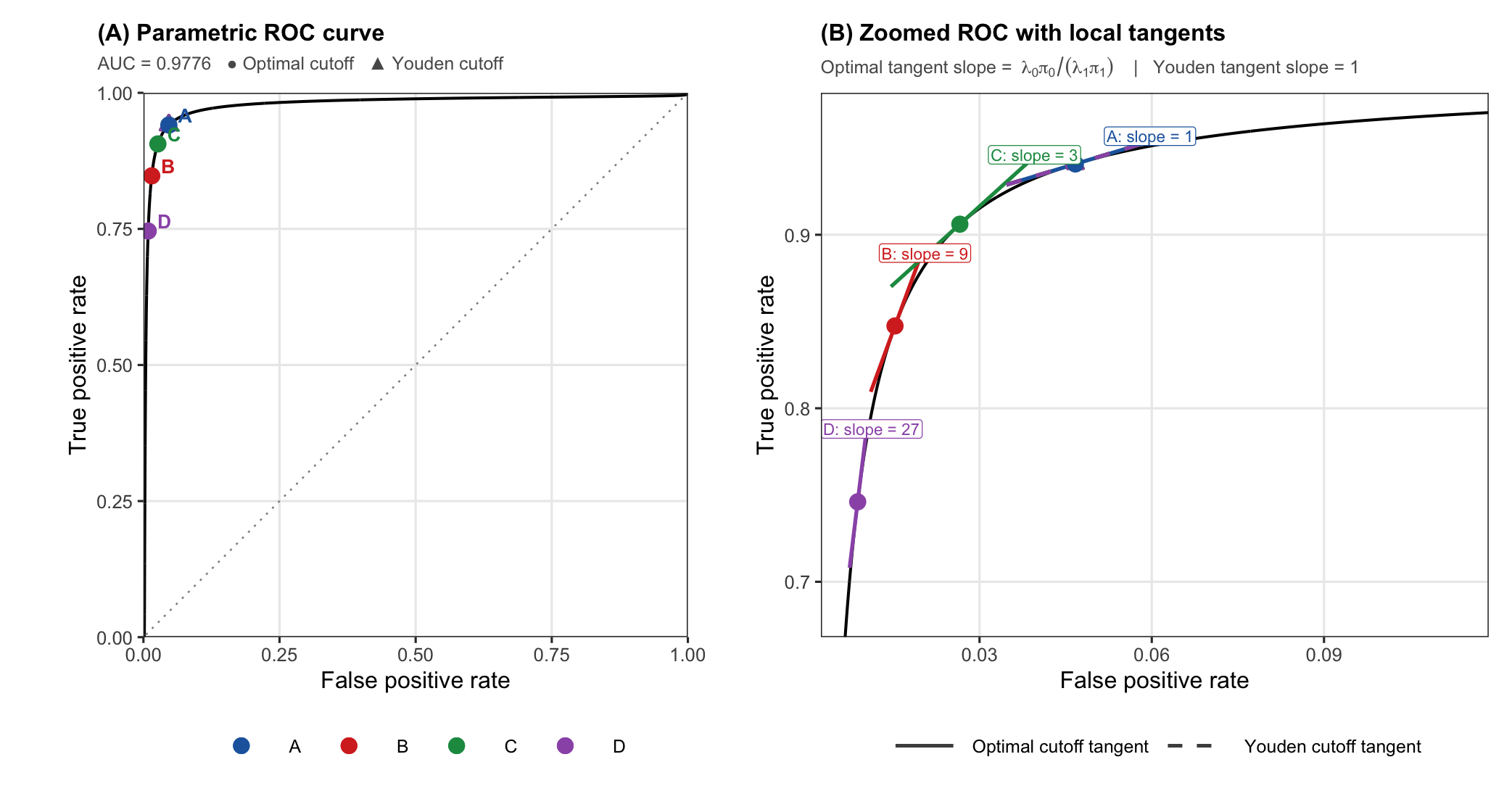}
\caption{
Parametric ROC curve for the RBD IgG biomarker under four decision-theoretic scenarios, with local tangent interpretations. Panel (A) shows the global ROC curve and the operating points associated with the optimal and Youden cutoffs. Panel (B) zooms into the relevant operating region and displays the local tangents at the optimal cutoffs, whose slopes equal $(\lambda_0\pi_0)/(\lambda_1\pi_1)$, together with the Youden tangent of slope 1. The figure provides a geometric interpretation of how prevalence and misclassification costs shift the optimal operating point away from the Youden solution.
}
\label{fig:rbd-roc-tangents}
\end{figure*}

Together, Figures~\ref{fig:rbd-densities} and~\ref{fig:rbd-roc-tangents} show that departures from the symmetric setting induce systematic movements of the optimal operating point along the ROC curve, which are then quantified numerically in Table~\ref{tab:RBD_decision}.

\begin{table*}[h!]
\centering
\caption{
Decision-theoretic optimal cutoffs for the RBD\_IgG antibody under four scenarios defined by different combinations of misclassification costs $(\lambda_0,\lambda_1)$ and prevalences $(\pi_0,\pi_1)$. 
The optimal cutoff $c^*$ minimizes the weighted misclassification risk, while $c_Y$ denotes the Youden cutoff. 
Results are reported on both the $\log_{10}$ scale and the original scale. 
We also report the optimal risk, the risk at the Youden cutoff, their difference ($\Delta$R), the asymptotic standard error (SE), Wald confidence intervals, and the identifiability quantity $|\partial_c \varphi|$.
\label{tab:RBD_decision}
}

\scriptsize
\setlength{\tabcolsep}{3pt}
\renewcommand{\arraystretch}{0.9}

\begin{tabular*}{\textwidth}{@{\extracolsep\fill}l c c c c c c c c c c c c @{\extracolsep\fill}}
\toprule
\textbf{Scen.} & $(\lambda_0,\lambda_1)$ & $(\pi_0,\pi_1)$ & $c^*$ & $c_Y$ & $c^*$ (orig.) & $c_Y$ (orig.) & $R(c^*)$ & $R(c_Y)$ & $\Delta$R & SE & CI$_{\text{orig}}$ & $|\partial_c \varphi|$ \\
\midrule

A & (1,1) & (0.5,0.5) & 2.2811 & 2.2811 & 191.03 & 191.03 & 0.0529 & 0.0529 & 0.0000 & 0.0019 & [189.40, 192.68] & 0.4004 \\

B & (1,1) & (0.9,0.1) & 2.7213 & 2.2811 & 526.39 & 191.03 & 0.0290 & 0.0479 & 0.0189 & 0.0028 & [519.87, 533.00] & 0.1279 \\

C & (1,3) & (0.9,0.1) & 2.4854 & 2.2811 & 305.75 & 191.03 & 0.0521 & 0.0598 & 0.0077 & 0.0021 & [302.87, 308.65] & 0.3066 \\

D & (3,1) & (0.9,0.1) & 3.0009 & 2.2811 & 1002.05 & 191.03 & 0.0491 & 0.1320 & 0.0829 & 0.0042 & [983.32, 1021.13] & 0.1542 \\

\bottomrule
\end{tabular*}

\begin{tablenotes}
\item {\it Note:} Differences between $c^*$ and $c_Y$ highlight the impact of asymmetric prevalence and misclassification costs on optimal decision thresholds. The quantity $|\partial_c \varphi|$ reflects local identifiability, with smaller values indicating greater sensitivity of the cutoff to parameter perturbations.
\end{tablenotes}

\end{table*}

\subsubsection{Scenario A: symmetric reference case $(\lambda_0=1,\lambda_1=1;\;\pi_0=0.5,\pi_1=0.5)$}
When costs and prevalences are equal, the optimality
condition reduces to $\Lambda(c^*;\theta)=1$, and
Proposition~\ref{prop:youden-relationship} guarantees
$\hat{c}^*=\hat{c}_Y$.
Both thresholds are equal to $2.281$ on the $\log_{10}$
scale, corresponding to $191.0$ antibody units on the
original scale.
The associated risk is $\hat{R}(\hat{c}^*)=0.0529$
and the confidence interval $[189.4,\,192.7]$ reflects
high precision.
This scenario provides the baseline against which the
effect of departures from symmetry is assessed.

\subsubsection{Scenario B: rare disease with equal costs $(\lambda_0=1,\lambda_1=1;\;\pi_0=0.9,\pi_1=0.1)$}
Introducing prevalence asymmetry while maintaining equal
costs modifies the target likelihood ratio to
\[
\frac{\lambda_0\pi_0}{\lambda_1\pi_1}
=
\frac{0.9}{0.1}
=9,
\]
substantially exceeding the unit value of Scenario~A.
By the monotone likelihood ratio property of
Proposition~\ref{prop:optimal-cutoff}, the optimal
threshold must shift to a region where $\Lambda(c;\theta)$ is
larger, i.e., further into the right tail of the
seronegative distribution.
The estimated cutoff $\hat{c}^*=2.721$ on the $\log_{10}$
scale corresponds to $526.4$ antibody units
(Figure~\ref{fig:rbd-densities}B), a substantial upward
displacement of $0.440$ on the $\log_{10}$ scale
relative to Scenario~A.
This shift is rational: in a population that is $90\%$
disease-free, the cost of false positives is amplified by
the high prevalence of the negative class, and the optimal
procedure accepts a higher false-negative rate in order
to reduce unnecessary positive classifications.
The risk at $\hat{c}^*$ is $0.0290$, compared to
$0.0479$ at the Youden cutoff, a reduction of
$\Delta\hat{R}=0.0189$ ($39\%$).
The local diagnostic $|\partial_c\varphi|=0.128$ is
the lowest across all four scenarios, reflecting the
scarcity of density from both distributions in the far
right tail; accordingly, the confidence interval
$[519.9,\,533.0]$ is substantially wider than in
Scenario~A.

\subsubsection{Scenario C: rare disease with higher
false-negative cost $(\lambda_0=1,\lambda_1=3;\;\pi_0=0.9,\pi_1=0.1)$}
Increasing the false-negative cost to $\lambda_1=3$
lowers the target likelihood ratio to
\[
\frac{\lambda_0\pi_0}{\lambda_1\pi_1}
=
\frac{0.9}{0.3}
=3,
\]
which lies between the target values of Scenarios~A
and~B.
Consistently with the ordering predicted by Proposition~\ref{prop:youden-relationship} and its first-order expansion~\eqref{eq:youden-expansion}, the
optimal cutoff $\hat{c}^*=2.485$ ($305.7$ antibody units)
is intermediate between the two:
\[
\hat{c}^*_A
=191.0, \quad \hat{c}^*_C =305.7, \quad \hat{c}^*_B=526.4.
\]
By shifting the threshold leftward relative to Scenario~B,
the procedure accepts more false positives in exchange
for reducing the costlier false negatives.
The risk reduction relative to the Youden cutoff is
$\Delta\hat{R}=0.0077$ ($13\%$), more modest than in
Scenario~B because the higher false-negative cost
partially counteracts the benefit of raising the threshold.
The diagnostic $|\partial_c\varphi|=0.307$ indicates a
moderately well-conditioned boundary, and the confidence
interval $[302.9,\,308.6]$ reflects good precision.

\subsubsection{Scenario D: rare disease with higher
false-positive cost $(\lambda_0=3,\lambda_1=1;\;\pi_0=0.9,\pi_1=0.1)$}
Assigning the highest cost to false positives drives the
target likelihood ratio to
\[
\frac{\lambda_0\pi_0}{\lambda_1\pi_1}
= \frac{2.7}{0.1}=27,
\]
the largest across all four scenarios.
The optimal cutoff $\hat{c}^*=3.001$ ($1002.0$ antibody
units) is located deep in the right tail of the
seronegative distribution (Figure~\ref{fig:rbd-densities}D),
reflecting the most conservative classification rule
considered.
The risk reduction relative to the Youden cutoff is
$\Delta\hat{R}=0.0829 (\,\approx 63\%)$, by far the largest among the
four scenarios: the Youden rule incurs a risk of
$0.1320$, more than twice the optimal $0.0491$.
This dramatic suboptimality of the symmetric threshold
under high false-positive cost underscores the practical
importance of the decision-theoretic approach.
The diagnostic $|\partial_c\varphi|=0.154$ and the wide
confidence interval $[983.3,\,1021.1]$ both reflect the
low local density at this extreme operating point.

\subsection{Cross-biomarker comparison and general remarks}

Table~\ref{tab:real-data-IgG-cutoffs-interval} shows a clear and
consistent pattern across the four IgG biomarkers analysed under the
asymmetric decision setting
$(\lambda_0,\lambda_1;\pi_0,\pi_1)=(1,3;0.9,0.1)$.
For every biomarker, the decision-theoretic optimal cutoff
$\hat c^*$ yields lower estimated risk than both the parametric
Youden cutoff $\hat c_Y$ and the empirical Youden cutoff
$\hat c_Y^{\mathrm{emp}}$.
Relative to the parametric Youden rule, the reduction in estimated
risk ranges from $0.0077$ for RBD\_IgG to $0.0155$ for S1\_IgG.
Relative to the empirical Youden rule, the reduction ranges from
$0.0029$ for Stri\_IgG to $0.0250$ for RBD\_IgG.
These gains illustrate the practical value of combining flexible SMSN
modelling with a decision rule that explicitly incorporates prevalence
and asymmetric misclassification costs.

The results also reveal substantial heterogeneity in the local
identifiability of the cutoff across biomarkers.
The diagnostic $|\partial_c\varphi(\hat c^*,\hat\theta)|$ ranges from
$0.3066$ for RBD\_IgG to $1.5701$ for S1\_IgG, with intermediate values
for Stri\_IgG ($0.3347$) and S2\_IgG ($0.7436$).
This variation is reflected directly in the precision of the estimated
thresholds: biomarkers with steeper estimating equations yield tighter
confidence intervals, whereas biomarkers with flatter local geometry
produce wider intervals.
In particular, S1\_IgG has the narrowest interval,
$[140.24,\,142.36]$, while RBD\_IgG and Stri\_IgG, which have the
smallest local slope values, display wider intervals,
$[302.86,\,308.64]$ and $[339.38,\,345.61]$, respectively.
This reinforces the interpretation of
$|\partial_c\varphi(\hat c^*,\hat\theta)|$ as a practically useful
measure of cutoff reliability.

More broadly, the application highlights the systematic effect of the
cost--prevalence ratio on threshold selection.
Since $\frac{\lambda_0\pi_0}{\lambda_1\pi_1}=3$,
the optimal cutoff is expected to exceed the Youden cutoff, and this
is precisely what is observed for all biomarkers.
Thus, the empirical findings are fully consistent with the theoretical
characterisation developed in Section~\ref{sec:cutoff}: under
low prevalence and higher false-negative cost, the decision-relevant
operating point is shifted away from the symmetric Youden solution
towards a more conservative threshold.
The comparison across biomarkers further shows that the practical
impact of this shift depends not only on the asymmetric decision
structure itself, but also on the local geometry of the fitted class
distributions around the decision boundary.

Finally, the numerical implementation proved stable for all four
biomarkers.
The admissible interval constructed from extreme empirical quantiles
of the pooled $\log_{10}$-transformed observations satisfied the
bracketing condition
$\varphi(a;\hat\theta)<0<\varphi(b;\hat\theta)$ in every case without further expansion.
This indicates that the proposed optimisation strategy reliably localised the relevant root of the estimating equation in the serological application.
Taken together, these results show that the Youden rule may be substantially misaligned with the actual decision problem under asymmetric prevalence and cost structures, and that the present decision-theoretic framework provides a principled and inferentially grounded alternative.

\section{Conclusions}\label{sec:discussion}

This paper developed a unified parametric framework for ROC analysis and optimal cutoff selection under the family of scale mixtures of skew-normal (SMSN) distributions, combining decision-theoretic threshold selection with formal asymptotic inference.

We characterised the optimal threshold as the solution of the likelihood-ratio equation that extends the classical Youden criterion
to settings with unequal disease prevalence and asymmetric
misclassification costs. Under a monotone likelihood ratio condition,
Proposition~\ref{prop:optimal-cutoff} established existence, uniqueness, and global optimality of the cutoff, while Proposition~\ref{prop:youden-relationship} quantified, through a first-order expansion, how the optimal threshold departs from the Youden solution as the cost--prevalence ratio moves away from the symmetric case. This provides a rigorous explanation for why the Youden rule may be suboptimal in practically relevant decision problems, especially when false negatives are more costly, or the target condition is rare.

We showed that the plug-in estimator $\hat c$ is consistent and asymptotically normal, with an explicit variance formula depending on two components: the estimation uncertainty of the model parameters and the local slope of the estimating equation at the optimal cutoff.
The quantity, $|\partial_c\varphi(c^*,\theta^*)|$, was interpreted as a local identifiability diagnostic: when the estimating equation is steep near the optimum, the cutoff is sharply defined and can be estimated with greater precision; when it is locally flat, the variance of $\hat c$ is inflated, and finite-sample behaviour becomes more delicate. This diagnostic arises naturally from the asymptotic theory and provides a practical tool for assessing the stability of the estimated threshold.

The simulation study based on six scenarios under skew-normal and skew-$t$ models provides strong empirical support for the theory. Across all scenarios, the estimator $\hat c$ displayed negligible bias and decreasing RMSE as the sample size increased. Moreover, the empirical variance of $\sqrt{n}(\hat c-c^*)$ was consistently close to the average plug-in asymptotic variance, with the ratio Var$_{\mathrm{emp}}$/Var$_{\mathrm{th}}$ remaining close to one throughout. Wald confidence intervals achieved coverage probabilities close to the nominal level in all settings considered. The most challenging configuration was the heavy-tailed weak-separation scenario, which exhibited smaller values of the local slope diagnostic, larger interval lengths, and more delicate numerical bracketing in smaller samples, exactly as predicted by the theoretical
analysis.

The application of the methodology developed in this paper to SARS-CoV-2 IgG biomarkers illustrates the practical value of the proposed framework. Using a skew-$t$ model for the seronegative group and a skew-normal model for the seropositive group, the optimal cutoff was estimated under the asymmetric decision setting $(\lambda_0,\lambda_1;\pi_0,\pi_1)=(1,3;0.9,0.1)$.
For all four biomarkers retained in the final analysis, the estimated
optimal threshold exceeded both the parametric and empirical Youden
cutoffs, in agreement with the theoretical implication of a
cost--prevalence ratio equal to three. More importantly, the optimal cutoff yielded uniformly lower estimated risk than either Youden-based alternative. Relative to the parametric Youden rule, the reduction in estimated risk ranged from $0.0077$ to $0.0155$; relative to the empirical Youden rule, it ranged from $0.0029$ to $0.0250$, with the largest gain observed for the RBD IgG biomarker. The local identifiability diagnostic also proved informative in practice, with biomarkers having smaller values of $|\partial_c\varphi|$ producing visibly wider confidence intervals.

\subsection*{Limitations and extensions}

The main theoretical results rely on a monotone likelihood ratio
assumption, which is natural and convenient but not automatic for the
full SMSN family when the two groups differ simultaneously in location,
scale, skewness, and tail behaviour. Thus, the MLR condition should be viewed as a sufficient condition for global uniqueness, not as a necessary one.
In practice, the numerical sign-change structure of $c\mapsto\varphi(c,\hat\theta)$ remains essential and should be checked
before reporting the estimated cutoff. The admissible-interval strategy adopted in this paper provides a stable approach for this verification.

In this paper, the prevalence weights $(\pi_0,\pi_1)$ were treated as fixed and known, so that the cutoff problem could be studied under a prespecified decision setting. In applications where prevalences must themselves be
estimated, their sampling variability would introduce an additional source of uncertainty into cutoff inference. Although the present framework could be extended to accommodate this case, a full treatment of estimated prevalences is beyond the scope of the current paper and is left for future work.

Although Wald inference performed well in the simulation settings considered here, weak-separation or heavy-tailed regimes may still challenge first-order approximations in smaller samples. Bootstrap confidence intervals, obtained by resampling independently within groups, provide a natural complementary approach and deserve a systematic comparison with Wald procedures in future work, especially under non-MLR configurations or near-flat estimating equations.

The numerical results reported here were obtained with a dedicated \textsf{R} implementation of plug-in cutoff estimation, asymptotic variance calculation, Wald inference, local identifiability diagnostics, and admissible-interval verification.


\subsection*{Author contributions}

Conceptualization: R.P., H.M., T.D.D. Methodology: R.P., H.M., T.D.D.
Software: R.P., T.D.D. Validation: R.P., H.M., T.D.D. Formal
analysis: R.P., T.D.D. Investigation: R.P., H.M., T.D.D. All authors
have read and agreed to the published version of the manuscript.

\subsection*{Funding}

This work was funded by national funds through FCT -- Fundação para a
Ciência e a Tecnologia, I.P., under CEAUL Research Unit,
UID/00006/2025, DOI:
\url{https://doi.org/10.54499/UID/00006/2025}.

\subsection*{Conflicts of interest}

The authors declare no conflicts of interest.

\subsection*{Data and code availability}

The data analysed in Section~\ref{sec:application} are from
\citet{rosado2020multiplex}. The \textsf{R} implementation of plug-in
cutoff estimation, asymptotic variance calculation, Wald inference,
local identifiability diagnostics, and admissible-interval verification
is available from the corresponding author upon reasonable request.

\bibliographystyle{abbrvnat}
\bibliography{references}

\end{document}